\documentclass[journal,comsoc]{IEEEtran}

\usepackage{mathrsfs}
\usepackage{amsmath,amsfonts,amssymb}
\usepackage{algorithmicx}

\usepackage{array}
\usepackage{subcaption}  
\usepackage{textcomp}
\usepackage{stfloats}
\usepackage{url}
\usepackage{verbatim}
\usepackage{graphicx}
\usepackage{booktabs}
\usepackage{balance}
\usepackage{color}
\usepackage{algorithm}
\usepackage{algpseudocode}   
\usepackage{cleveref}       

\def\BibTeX{{\rm B\kern-.05em{\sc i\kern-.025em b}\kern-.08em
    T\kern-.1667em\lower.7ex\hbox{E}\kern-.125emX}}

\newtheorem{lemma}{Lemma}
\newtheorem{corollary}{Corollary}
 
\title{A Unified Codebook Design for Curvature-Reconfigurable Apertures: Seamless Near to Far Field Coverage}

\author{Zhoujie You, Shu Sun, Ruifeng Gao, Jue Wang, and Xianghao Yu

\thanks{Zhoujie You and Ruifeng Gao are with the School of Transportation and Civil Engineering, Nantong University, Nantong 226019, China (e-mail: zjyou@stmail.ntu.edu.cn; grf@ntu.edu.cn).}
 
\thanks{Shu Sun is with the School of Information Science and Electronic Engineering, Shanghai Jiao Tong University, Shanghai 200240, China (e-mail: shusun@sjtu.edu.cn).}
\thanks{Jue Wang is with the School of Information Science and Technology, Nantong University, Nantong 226019, China (e-mail: wangjue@ntu.edu.cn).}
\thanks{Xianghao Yu is with the Department of Electrical Engineering, City University of Hong Kong, Hong Kong, China (e-mail: alex.yu@cityu.edu.hk).}

}

\begin{document}

\maketitle

\begin{abstract}
Beam training for extremely large-scale arrays with curvature-reconfigurable apertures (CuRAs) faces the critical challenge of severe, geometry-dependent angle-range coupling. While most existing designs compartmentalize near field and far field scenarios, we propose a unified, distance-adaptive hierarchical codebook framework for 1-D and 2-D CuRAs that seamlessly bridges both propagation regimes. Under a spherical-wave model, we first characterize the beamforming-gain correlation in a polar angular domain, deriving an angle-dependent angular sampling rule to capture the varying curvature. To achieve full-range coverage, we introduce a direction-dependent effective Rayleigh distance (ERD) as a soft boundary to gate the range sampling. Crucially, by sampling uniformly in the reciprocal-range domain, the proposed codebook provides precise, dense focusing within the ERD and automatically degenerates into sparse, angle-only steering beyond it. This mechanism eliminates the need for hard mode-switching between near- and far-field operations. Simulation results demonstrate that our unified design consistently outperforms representative baselines in spectral efficiency and alignment accuracy, offering a comprehensive solution for full-range CuRA communications.

\end{abstract}


\section{Introduction}

\IEEEPARstart{I}{n}  6G wireless systems, extremely large antenna arrays (ELAAs) are a key enabler of high-rate and low-latency transmission because of their large array gain, high spatial resolution, and spectral-efficiency benefits \cite{6G, LYW}. As array apertures grow, the Fraunhofer distance can exceed practical link distances, pushing users into the radiative near field \cite{elaa1,elaa,fresnel,fnear}. Under such conditions, the conventional planar-wave assumption no longer holds, and spherical wavefronts must be accounted for in channel modeling and beamforming \cite{region,nfc,nfc1}. As a result, classical far field designs for channel estimation, precoding, and beam training need to be reexamined for near field ELAA systems \cite{turntonfc,turntonfc2,turntonfc3}. In particular, codebook-based beam alignment has emerged as a promising approach, as it enables near field energy focusing with limited training overhead while alleviating the burden of acquiring instantaneous CSI in massive-array, time-varying environments \cite{beam,cdbook,csi}.

To cope with the new beamforming and training requirements in the near field, a variety of array architectures and corresponding transceiver designs have been investigated \cite{focusing,focusing2}. Among them, the uniform linear array (ULA) is one of the most representative configurations for ELAA systems \cite{ula1}. Existing works have examined the distinctive characteristics of near field channel modeling, estimation, and beam training under ULA deployments. For example, \cite{ula} highlights the fundamental differences between far field and near field channel modeling and estimation, while \cite{ula2} develops channel estimation methods tailored to ULA-based near field systems. In addition, \cite{training} proposes an optimization-based near field codebook together with a multi-beam training strategy to enhance interference suppression in multi-user scenarios.

Despite their analytical simplicity, ULAs may suffer from pronounced angle-dependent degradation in near field focusing performance, especially at large incidence angles, which limits the user region that can effectively benefit from near field beam focusing \cite{enable}. This observation has motivated the investigation of circular array geometries. Owing to their rotational symmetry in the azimuth plane, uniform circular arrays (UCAs) can provide more uniform azimuthal coverage and are attractive for near field beam training. For instance, \cite{uca} develops an efficient beam-training scheme for terahertz near field communications that reduces the required codebook size, while \cite{uca2} shows that UCAs can outperform ULAs by generating Bessel-like beams that enhance spatial focusing and suppress interference. \cite{training2} further proposes a non-uniform distance sampling strategy and corresponding codebook design, together with a two-stage search procedure, to enable efficient three-dimensional (3-D) beam focusing. Extending circular apertures to 3-D, cylindrical arrays can support joint focusing in both the horizontal and vertical dimensions \cite{cylin}, and this capability has been further exploited in \cite{cla} for near field codebook design with extremely large cylindrical arrays.

However, most existing near field studies are built on fixed-geometry apertures, such as ULAs, UCAs, and cylindrical arrays, and do not explicitly account for platform-imposed conformality constraints or exploit array curvature itself as a controllable design degree of freedom. Such limitations are particularly relevant for practical platforms including vehicle/vessel bodies, UAV fuselages, and curved building facades \cite{UAV,Maritime}, where the array geometry is often coupled with the physical surface on which it is deployed. In these scenarios, a mismatch between the aperture shape and the incident wavefront can degrade beam-focusing performance \cite{de}. Reconfigurable conformal arrays offer a promising alternative, as their physical curvature can be adjusted to better match near field propagation characteristics while maintaining a regular element layout, thereby improving 3-D focusing capability and enabling more effective joint distance-angle control \cite{caa,aulaupa}. Different from conventional conformal arrays with fixed curvature, we treat curvature as a tunable design parameter that directly reshapes the near field focusing pattern and, hence, the codebook design space. This perspective is particularly important because far field codebooks cannot provide distance-selective focusing \cite{polarcd}, whereas many existing near field codebooks rely on dense discretization of the 3-D distance--angle space, leading to rapidly increasing codebook size, training overhead, and computational complexity as the desired focusing resolution improves \cite{3Ds}. To the best of our knowledge, the joint design of curvature-reconfigurable apertures (CuRAs) and training-efficient near field codebooks remains largely unexplored.

Motivated by the above observations, we investigate near field communications enabled by CuRAs and develop a training-efficient codebook design tailored to such curved apertures. Specifically, we consider two representative architectures, namely, the one-dimensional CuRA (1-D CuRA) and the two-dimensional CuRA (2-D CuRA), in which the aperture curvature is treated as a controllable structural parameter. Both architectures naturally encompass straight and curved configurations as special cases, thereby providing a unified framework for geometry-adaptive near field beam training.

Starting from the 1-D CuRA, we characterize the angle--range-coupled near field beam response induced by the curved aperture geometry. Based on this observation, we develop a polar-domain codebook design framework together with an effective Rayleigh distance (ERD)-guided hierarchical construction for efficient near field beam training with reduced codebook size and overhead. The proposed design is further extended to the 2-D CuRA, leading to a structured hierarchical codebook for 3-D beam focusing. The main contributions of this paper are summarized as follows:
\begin{itemize}
\item We investigate the spherical-wave propagation for CuRAs and analytically reveal how the aperture curvature reshapes the angular resolution and range sensitivity. This characterization highlights the fundamental, direction-dependent angle-range coupling induced by the reconfigurable geometry.

\item We propose a hierarchical, polar-domain codebook framework for the 1-D CuRA. By introducing the direction-dependent ERD, we develop an ERD-guided reciprocal-range sampling rule. This mechanism provides dense, distance-aware focusing within the ERD and automatically degenerates into sparse, angle-only steering beyond it, enabling seamless full-range coverage.

\item We extend the proposed unified framework to 2-D CuRAs by developing separable yet structured sampling rules across both the linear and curved aperture dimensions. This yields a compact, hierarchical 3-D codebook that efficiently manages the anisotropic focusing behavior of the 2-D array.

\item Extensive simulation results validate our theoretical analysis and demonstrate the superiority of the proposed CuRA-oriented hierarchical codebook. Compared to representative baselines, our unified design achieves reliable full-range beam coverage with a reduced codebook size, significantly lowering training overhead while improving overall spectral efficiency.
\end{itemize}

The remainder of this paper is organized as follows. Section~\ref{sec:analysis} presents the system model for the 1-D CuRA and characterizes its near field beamfocusing vectors. Section~\ref{sec:sampling} then develops a spherical-domain codebook for the 1-D CuRA by sampling the near field focusing vectors. Building on this design, Section~\ref{sec:arcUPA} extends the proposed framework from the 1-D CuRA to the 2-D CuRA. Section~\ref{sec:simulation} provides simulation results, and Section~\ref{sec:conclusion} concludes the paper.

\section{System Model of the 1-D CuRA}  \label{sec:analysis}

\begin{figure}[!t]
    \centering 
   \includegraphics[width=0.4\textwidth, keepaspectratio]{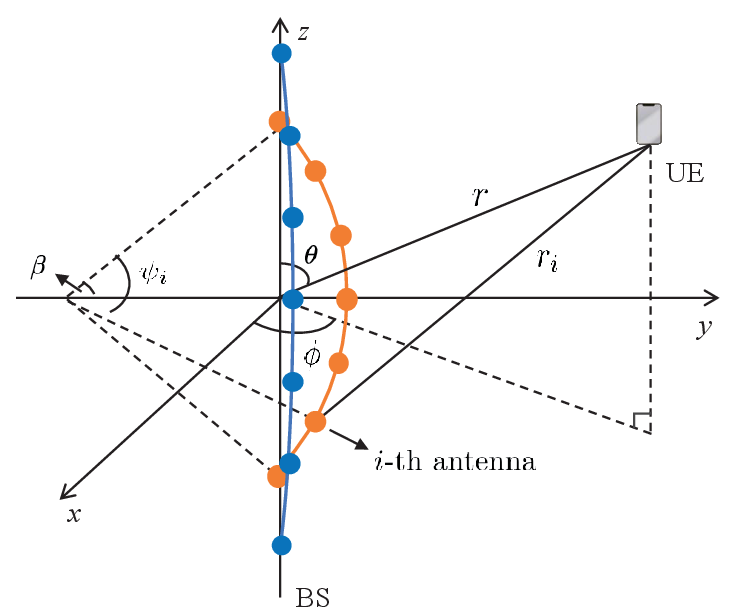 }
    \caption{System model of the 1-D CuRA.} 
    \label{fig:system2} 
        \vspace{-1.5 em}
\end{figure}

We consider a downlink single-user single-stream scenario where a base station (BS) equipped with a 1-D CuRA of $N$ antenna elements serves a single-antenna user, as shown in Fig.~\ref{fig:system2}. The BS employs analog beamforming and performs beam alignment using a finite beam codebook. The 1-D CuRA is obtained by assigning a curvature-reconfigurable geometry to a one-dimensional aperture with regular element indexing, which reduces to a straight aperture as a special case.

\subsection{Geometry of the 1-D CuRA}
As shown in Fig.~\ref{fig:system2}, the 1-D CuRA is deployed on the $yz$-plane as a circular-arc aperture with curvature radius $R$ and central angle $2\beta\in(0,\pi]$, yielding an arc length $L=2\beta R$. With uniform element spacing $d=\lambda/2$ along the arc, we have $L=(N-1)d$, and thus $\beta = {(N-1)d}/{2R}$. Under a 3-D Cartesian coordinate system, the location of the $i$-th antenna element ($i=1, \ldots, N$) on the circular arc is $\mathbf q_i=\Big(0,\;R\big(\cos(\beta-\psi_i)-\cos\beta\big),\;R\sin(\beta-\psi_i)\Big)$, where the arc parameter $\psi_i$ for uniform element spacing is $\psi_i = {2\beta(i-1)}/{(N-1)}$.

The user location is expressed in spherical coordinates as $\mathbf u(r,\theta,\phi)=\big(r\sin\theta\cos\phi,\;r\sin\theta\sin\phi,\;r\cos\theta\big)$, 
where $r>0$ denotes the link distance, $\theta\in[0,\pi]$ is the elevation angle measured from the positive $z$-axis, and $\phi\in[0,\pi)$ is the azimuth angle measured counterclockwise from the positive $x$-axis on the $xy$-plane. 

Based on the above geometry, the propagation distance between the user and the $i$-th antenna element is
\begin{equation}
r_i=\big\|\mathbf u(r,\theta,\phi)-\mathbf q_i\big\|_2, ~ i=1,\ldots,N,
\label{eq:ri_def}
\end{equation}
which directly reflects the spherical-wave propagation in the near field. In the next subsection, we use $\{r_i\}_{i=1}^N$ to construct the near field beamfocusing (array-response) vector and to establish the downlink channel model for the 1-D CuRA.

\subsection{Near field Downlink Channel Model of the 1-D CuRA}
We adopt a spherical-wave channel model for the 1-D CuRA. The downlink channel between the BS and the user is denoted by $\mathbf h\in\mathbb C^{N\times 1}$. A general near field multipath channel can be written as
\begin{equation}
\mathbf h=\sum_{\ell=1}^{L_p}\alpha_\ell\,\mathbf b(r_\ell,\theta_\ell,\phi_\ell),
\label{eq:multipath_channel}
\end{equation}
where $L_p$ is the number of propagation paths, $\alpha_\ell\in\mathbb C$ is the complex gain of the $\ell$-th path, and $\mathbf b(r_\ell,\theta_\ell,\phi_\ell)\in\mathbb C^{N\times 1}$ denotes the near field beamfocusing vector associated with the spherical coordinates $(r_\ell,\theta_\ell,\phi_\ell)$.
In this paper, we focus on LoS-dominant near field links, which are typical in mmWave/THz systems. For analytical tractability, we set $L_p=1$, so that \eqref{eq:multipath_channel} reduces to $\mathbf h=\alpha\,\mathbf b(r,\theta,\phi)$,

where $\alpha$ absorbs both large-scale attenuation and small-scale fading of the dominant path.

The near field beamfocusing vector $\mathbf b(r,\theta,\phi)$ compensates the spherical-wave phase variations across the array aperture and is defined as
\begin{equation}
\begin{aligned}  
\mathbf b(r,\theta,\phi)=\frac{1}{\sqrt N}\Big[&e^{-j\frac{2\pi}{\lambda}(r_1-r)},\;\ldots,\;e^{-j\frac{2\pi}{\lambda}(r_N-r)}\Big]^{\text{T}},
\label{eq:bf_vector}
\end{aligned}
\end{equation}
where $r_i$ is the exact distance defined in \eqref{eq:ri_def} and its simplified form can be expressed as \eqref{dis}, where the approximation (a) is derived from the second-order Taylor series expansion $\sqrt{1+x} = 1 + \frac{x}{2} - \frac{x^2}{8} + \mathcal{O}(x^3)$. It is worth emphasizing that $\mathbf b(r,\theta,\phi)$ depends jointly on the distance and the angular parameters, in contrast to far field steering vectors that depend only on direction. This intrinsic 3-D dependence enables spatial power focusing in the radiative near field and motivates codebook designs that explicitly account for the curved geometry of the 1-D CuRA.

\begin{figure*}[t]
\normalsize
\begin{equation}
\begin{aligned}
  r_{i}&=\sqrt{[ r\sin(\theta) \cos(\phi)]^2+[r\sin(\theta)\sin(\phi)-R\cos(\beta-\psi_i)+R\cos(\beta)]^2+[r\cos(\theta)-R\sin(\beta-\psi_i)]^2}\\
    &\mathop {\approx }\limits ^{(a)}  r - R [\sin\theta\sin\phi( \cos( \beta - \psi_i)-\cos\beta) +\cos\theta \sin( \beta - \psi_i)] \\  
    & \quad\quad+ \frac{R^2}{2r} \left[ 1 - 2\cos(\beta-\psi_i)\cos\beta + \cos^2\beta \right.  \left. - (\sin\theta\sin\phi(\cos(\beta-\psi_i)-\cos\beta) + \cos\theta\sin(\beta-\psi_i))^2 \right] , \label{dis}
\end{aligned}
\end{equation}
\hrulefill 
\end{figure*}

Let $x$ denote the transmitted symbol with $\mathbb E[|x|^2]=P$. The BS applies an analog beamforming vector $\mathbf v\in\mathbb C^{N\times 1}$ satisfying $\|\mathbf v\|_2=1$. Under phase-shifter implementation, $\mathbf v$ further satisfies the constant-modulus constraint $|[\mathbf v]_i|=1/\sqrt{N}$.
The received signal is
\begin{equation}
y = \mathbf h^{\mathrm H}\mathbf v\,x+n_0
  = \alpha\,\mathbf b^{\mathrm H}(r,\theta,\phi)\mathbf v\,x+n_0,
\end{equation}
where $n_0\sim\mathcal{CN}(0,\sigma^2)$ is additive white Gaussian noise.

For the dominant-path channel and constant-modulus analog beamforming, the received power is maximized by matched-filter transmission, i.e., $\mathbf v^* = e^{j\varphi}\mathbf b(r,\theta,\phi)$, 
where $e^{j\varphi}$ is an arbitrary common phase rotation. Therefore, beam alignment aims to select a beam $\mathbf v$ from a finite codebook $\mathcal{W}$ to maximize the beamforming gain
\begin{equation}
\max_{\mathbf v\in\mathcal W}\; \big|\mathbf b^{\mathrm H}(r,\theta,\phi)\mathbf v\big|^2.
\end{equation}
The training overhead is measured by the number of probed beams during the alignment procedure.

\section{Codebook Design for the 1-D CuRA} \label{sec:sampling}

As discussed in Section~\ref{sec:analysis}, acquiring accurate instantaneous CSI for near field 1-D CuRA links is challenging. The near field array response depends jointly on direction and range, and the resulting channel is generally not sparse in the conventional angular domain, which makes high-resolution channel estimation and feedback costly in both training overhead and computational complexity. To avoid explicit CSI acquisition, we adopt codebook-based beam training as a practical approach.

In this section, we design a dedicated codebook $\mathcal{C}=\{\mathbf w_m\}_{m=1}^{M}$ for the 1-D CuRA. The goal is to ensure that for any user location $\mathbf s=(r,\theta,\phi)$ within a target service region $\mathcal S$, there exists at least one codeword $\mathbf w_m\in\mathcal C$ that provides a sufficiently large beamforming gain. Therefore, we formulate the codebook design as a correlation-based coverage optimization to guarantee worst-case performance. Given a predefined correlation threshold $\delta \in (0, 1)$, the codebook must satisfy as
\begin{equation}
\min_{\mathbf{s} \in \mathcal{S}} \; \max_{\mathbf{w} \in \mathcal{C}} \; \left| \mathbf{b}^{\mathrm{H}}(\mathbf{s}) \mathbf{w} \right| \geq \delta,
\label{eq:design_criterion}
\end{equation}
where $\mathbf b(\mathbf s)$ is the near field beamfocusing vector in \eqref{eq:bf_vector}. With this definition, \eqref{eq:design_criterion} suggests a geometric coverage interpretation: the service region $\mathcal S$ should be covered by a set of beam neighborhoods centered at representative locations $\{\mathbf s_m\}_{m=1}^{M}$, where each center corresponds to a codeword, and the neighborhood size is controlled by $\delta$.

To facilitate the design, we measure the similarity between two beams targeting locations $\mathbf s_1$ and $\mathbf s_2$ through
\begin{equation}
G(\mathbf s_1,\mathbf s_2) \triangleq \big|\mathbf b^{\mathrm H}(\mathbf s_1)\mathbf b(\mathbf s_2)\big|\in[0,1].
\label{eq:correlation_definition}
\end{equation}

Accordingly, codebook construction can be cast as a sampling problem over the 3-D location space $(r,\theta,\phi)$. By finding sampling points $\{\mathbf s_m\}_{m=1}^{M}$ such that for any $\mathbf s\in\mathcal S$, there exists at least one $\mathbf s_m$ satisfying $G(\mathbf s,\mathbf s_m)\ge \delta$. However, the exact expression of $G(\cdot,\cdot)$ is difficult to manipulate due to the spherical-wave phase terms. To obtain a tractable design, we leverage the Fresnel (second-order) expansion of the propagation distance, which approximately separates an angle-dominant linear term from a range-dependent quadratic term. This separation enables a hierarchical design: we first construct a polar-domain angular grid over $(\theta,\phi)$, then characterize its sampling resolution, and finally refine the codebook via range sampling conditioned on each angular grid point.

\subsection{Polar Coordinate Transformation}\label{sec:polar}

Due to the rotational symmetry of the arched array around its normal axis, uniform sampling in the spherical coordinates $(\theta,\phi)$ does not lead to uniform angular resolution with respect to the array geometry, and it tends to oversample directions close to the array normal and undersample directions toward the edges. To better exploit this symmetry, we re-parameterize the direction using polar coordinates $(\rho,\varphi)$ defined with respect to the array normal axis. Let $\rho$ denote the angular deviation from the array normal (the $x$-axis), and let $\varphi$ denote the rotation angle around the normal axis. The transformation from $(\theta,\phi)$ to $(\rho,\varphi)$ is given by
\begin{equation}
\rho = \arccos(\sin\theta\cos\phi), ~
\varphi = \arctan\!\left(\frac{\sin\theta\sin\phi}{\cos\theta}\right)
\end{equation}
Here, $\rho=0$ corresponds to the broadside direction (normal to the array), and larger $\rho$ indicates a larger deviation from broadside. In the following, we derive an approximate expression of the beamforming gain in the proposed polar domain.

\subsection{Analysis of Sampling Method for 1-D CuRA}
\subsubsection{ Joint Angular Domain }
We first characterize how the beamforming gain varies with the departure direction. To isolate the angular effect, we consider two candidate focusing directions at the same link distance $r$, i.e.,
$\mathbf s_1=(r,\rho_1,\varphi_1)$ and $\mathbf s_2=(r,\rho_2,\varphi_2)$ in the proposed polar coordinate system. We evaluate the angular correlation between the corresponding near field focusing vectors and adopt a first-order approximation with respect to $1/r$ to obtain a tractable expression for the angular resolution.

\begin{lemma}\label{lemma1}
For a fixed distance $r$, the angular-domain beamforming gain between two directions $(\rho_1,\varphi_1)$ and $(\rho_2,\varphi_2)$ can be approximated as
\begin{equation}
    \begin{aligned}
        G_{\text{angle}}(\mathbf{s}_1, \mathbf{s}_2) &\approx|\mathbf{a}^H(\rho_1,\varphi_1)\mathbf{a}(\rho_2,\varphi_2)|\\
&\approx \left| J_0\left( \frac{2\pi R}{\lambda}\cdot \Theta(\rho,\varphi) \cdot d_p \cdot\frac{\sin\beta}{\beta} \right) \right|,\label{eq:G_a}
    \end{aligned}
\end{equation}
where $\mathbf a(\rho,\varphi)$ denotes the far field steering vector obtained by applying a first-order Taylor expansion of the near field focusing vector $\mathbf b(r,\rho,\varphi)$ with respect to $1/r$. Moreover,$\Theta(\rho,\varphi)  = \sqrt{\cos^2\rho\sin^2\varphi + \sin^2\rho\cos^2\varphi}$, $d_p = \sqrt{ \Delta\rho^2 + \Delta\varphi^2 }$ with $\Delta\rho = \rho_2 - \rho_1$ and $\Delta\varphi = \varphi_2 - \varphi_1$.
\end{lemma}

\textit{Proof:} See Appendix \ref{app:proof_BG1}. \hfill$\blacksquare$\par

Lemma~\ref{lemma1} indicates that, in the polar domain, the angular correlation is governed mainly by the polar-domain separation $d_p$, which is up to a direction-dependent factor $\Theta(\rho,\varphi)$. As a result, the constant-correlation contours in $(\rho,\varphi)$ are approximately circular, which motivates constructing the angular codebook via a regular polar grid rather than optimizing an irregular grid in $(\theta,\phi)$. We determine the angular sampling interval by controlling the correlation between neighboring beams on the angular grid. Specifically, for any two adjacent grid points $\mathbf s_1\neq \mathbf s_2$ in $\mathcal{G}_{\text{angle}}$, we impose
\begin{equation}
    G_{\text{angle}}(\mathbf{s}_1, \mathbf{s}_2) \leq \delta_p, \quad \forall \mathbf{s}_1 \neq \mathbf{s}_2 \in \mathcal{G}_{\text{angle}},
\end{equation}
where $\delta_p\in(0,1)$ is a prescribed coherence threshold. Together with \eqref{eq:G_a}, this constraint yields an explicit mapping from the target correlation $\delta_p$ to the polar-grid step size $d_p$, enabling a direct selection of the angular sampling interval. A smaller $\delta_p$ reduces mutual coherence and is beneficial for recovery-oriented tasks such as compressed sensing, at the expense of a denser grid and a larger codebook. 

To control the correlation between neighboring beams, we select the polar-grid step size such that the Bessel-term argument in \eqref{eq:G_a} reaches the prescribed threshold. Specifically, let $x_{\delta}$ denote the smallest positive solution to $|J_0(x_{\delta})|=\delta_p$.\footnote{Equivalently, $x_{\delta}=J_0^{-1}(\delta_p)$ when restricting to the principal branch on $(0,x_{0,1})$, where $x_{0,1}$ is the first positive zero of $J_0(\cdot)$.} 
Imposing $G_{\text{angle}}(\mathbf s_1,\mathbf s_2)\approx \delta_p$ for neighboring directions yields the following local radial step size (with $\Delta\varphi=0$)

\begin{equation}
\Delta \rho
= \frac{J_0^{-1}(\delta_p)\,\lambda}{2\pi R}\cdot \frac{1}{\Theta(\rho,\varphi)}\cdot \frac{\beta}{\sin\beta}.
\label{eq:delta_rho}
\end{equation}
Accordingly, the radial grid points are given by
$\rho_i = i\,\Delta \rho,\quad i=0,1,\ldots,I_{\max},$
where $\rho_{\max}$ is the maximum radial angle to be covered and $I_{\max}=\left\lfloor \rho_{\max}/\Delta\rho \right\rfloor$.

For each radial layer $i$ with $\rho_i>0$, we determine the azimuthal spacing $\Delta\varphi$ such that the correlation between neighboring beams on the same ring does not exceed the prescribed threshold. Consider two adjacent points on the ring at radius $\rho_i$ with the same $\rho_i$ and separation $\Delta\varphi$. Under the local polar-domain approximation, their polar-domain distance satisfies $d_p \approx \rho_i\Delta\varphi$. Substituting this into \eqref{eq:G_a} and enforcing $G_{\text{angle}}\le \delta_p$ yields
\begin{equation}
\rho_i \Delta\varphi \leq \frac{J_0^{-1}(\delta_p)\lambda}{2\pi R}\cdot \frac{1}{\Theta(\rho_i,\varphi)}\cdot \frac{\beta}{\sin\beta},
\label{eq:angular_spacing_condition}
\end{equation}
where $\rho_i\Delta\varphi$ represents the arc-length approximation on the ring at radius $\rho_i$ in the polar domain. To obtain a uniform and conservative design, we adopt the worst-case bound $\Theta(\rho_i,\varphi)\le \Theta_{\max}=1$, which leads to the maximum allowable azimuthal spacing
\begin{equation}
\Delta\varphi(\rho_i) \leq \frac{J_0^{-1}(\delta_p)\lambda}{2\pi R\,\rho_i}\cdot \frac{\beta}{\sin\beta},\qquad \rho_i>0.
\label{eq:angular_spacing_max}
\end{equation}

Let the azimuthal search region in the polar domain be $\varphi\in[-\varphi_{\max},\varphi_{\max}]$. For the $i$-th radial layer with $\rho_i>0$, the required number of azimuth samples is
\begin{equation}
N_{\varphi}(\rho_i)
= \left\lceil \frac{2\varphi_{\max}}{\Delta\varphi(\rho_i)} \right\rceil
= \left\lceil \frac{4\pi R \rho_i}{\lambda\, J_0^{-1}(\delta_p)} \cdot \sin\beta \cdot \frac{\varphi_{\max}}{\beta} \right\rceil,
\label{eq:num_angular_samples}
\end{equation}
where $\lceil\cdot\rceil$ denotes the ceiling operator. In particular, when $\varphi_{\max}=\beta$, \eqref{eq:num_angular_samples} reduces to
$N_{\varphi}(\rho_i)=\left\lceil \frac{4\pi R \rho_i}{\lambda\, J_0^{-1}(\delta_p)} \sin\beta \right\rceil$.
For $\rho_0=0$, only one sample is required since $\varphi$ is undefined at the origin.

For each grid point $(\rho_i,\varphi_{ij})$, we construct the corresponding codeword by mapping it back to the standard spherical direction and then forming the far field steering vector. Specifically, using the inverse mapping in Section~\ref{sec:polar}, the associated spherical angles $(\theta_{ij},\phi_{ij})$ are obtained as
$\theta_{ij}=\arccos\!\big(\cos\rho_i\cos\varphi_{ij}\big)$ and $
\phi_{ij}=\arctan\big(\frac{\sin\rho_i\sin\varphi_{ij}}{\cos\rho_i}\big)$, with $\theta_{ij}\in[0,\pi]$ is the polar angle (measured from the positive $z$-axis) and $\phi_{ij}$ is the azimuth angle on the $xy$-plane. Given $(\theta_{ij},\phi_{ij})$, we compute the far field steering vector $\mathbf a(\theta_{ij},\phi_{ij})$ and normalize it to obtain the codeword as $\mathbf w_{ij}=\frac{\mathbf a(\theta_{ij},\phi_{ij})}{\|\mathbf a(\theta_{ij},\phi_{ij})\|_2}$.
Repeating the above construction for all grid points $(\rho_i,\varphi_{ij})\in\mathcal G_{\text{angle}}$ yields the angular-domain codebook. It is convenient to collect the resulting codewords into a matrix whose columns are the individual beams, i.e.,
\begin{equation}
\mathbf {W}_{\text{angle}}
=
\big[\,\mathbf w_{00},\{\mathbf w_{ij}\}_{i=1,\ldots,I_{\max},\,j=1,\ldots,N_{\varphi}(\rho_i)}\,\big]
\in\mathbb C^{N\times M_{\text{total}}},
\label{eq:codebook_matrix}
\end{equation}
where the total number of codewords is
\begin{equation}
M_{\text{total}}=1+\sum\nolimits_{i=1}^{I_{\max}} N_{\varphi}(\rho_i).
\end{equation}
This matrix representation enables compact storage and simplifies beam training, since the receiver can scan the columns of $\mathbf W_{\text{angle}}$ to identify the beam that maximizes the beamforming gain.

\subsubsection{ Distance Domain }
\label{BULA dis}

With the polar-domain angular grid $\mathcal G_{\text{angle}}$ established, we next design the range sampling along each fixed direction $(\rho,\varphi)\in\mathcal G_{\text{angle}}$. Similar to the angular-domain analysis, we characterize the beamforming gain between two near field focusing vectors that share the same direction but have different ranges. The resulting range-domain correlation admits the following approximation.

\begin{lemma}\label{lem:range_corr}
For two focusing vectors pointing to the same direction $(\rho,\varphi)$ but different ranges $r_1\neq r_2$, the beamforming gain can be approximated as
\begin{equation}
\begin{aligned}
G_r(\mathbf s_1,\mathbf s_2)
&= \big| \mathbf b^{\mathrm H}(r_1,\rho,\varphi)\,\mathbf b(r_2,\rho,\varphi) \big|\approx \big| J_0(\zeta) \big|,
\end{aligned}
\label{eq:Gr}
\end{equation}
where $\zeta = \frac{\pi R^2}{\lambda}\,\xi(\rho,\varphi,\beta)\,\Delta\tau,~\Delta\tau=\left|\frac{1}{r_2}-\frac{1}{r_1}\right|$. The factor $\xi(\rho,\varphi,\beta)$ quantifies the sensitivity of near field focusing to distance variations along direction $(\rho,\varphi)$ and is defined as
\begin{equation}
\xi(\rho,\varphi,\beta)
= \sqrt{ f_1(\rho,\varphi,\beta)^2 + f_2(\rho,\varphi,\beta)^2 },
\label{eq:xi_def}
\end{equation}
with $f_1(\rho,\varphi,\beta)=\sin\rho\,\sin\varphi\,\sin\beta-\cos\rho\,\cos\beta$, and
$f_2(\rho,\varphi,\beta)=\cos\rho\,\cos\varphi\,\sin\beta$.
\end{lemma}

\textit{Proof:} See Appendix \ref{app:proof_BG3}. \hfill$\blacksquare$\par

Lemma~\ref{lem:range_corr} characterizes the distance domain focusing behavior of near field beamfocusing vectors. In particular, the correlation depends on both the curvature-dependent factor $\xi(\rho,\varphi,\beta)$ and the reciprocal-range separation $\Delta\tau$, whereas in the far field the array response is range-invariant. This observation motivates defining an operational boundary that specifies when a near field model is necessary. A classical reference is the Rayleigh distance (RD) $r_{\mathrm R}=2D^2/\lambda$, which is derived from a $\pi/8$ phase-error criterion under the plane-wave approximation \cite{rrd}, \cite{pi/8}. However, using $r_{\mathrm R}$ alone for beamfocusing codebook design can be overly conservative and may lead to oversampling, because it does not capture the direction-dependent effective aperture of conformal arrays. To address this limitation, we adopt an ERD as a gain-loss-based near field boundary \cite{ula}.

From \eqref{eq:Gr}, letting $r_2\to\infty$ yields
$G_r(\mathbf s_1,\mathbf s_2)\approx \mu(r,\rho,\varphi)\triangleq \big|\mathbf b^{\mathrm H}(r,\rho,\varphi)\,\mathbf a(\rho,\varphi)\big|$, which represents the normalized correlation between the near field beam $\mathbf b$ and its far field counterpart $\mathbf a$. We quantify the beamforming gain loss by $1-\mu(r,\rho,\varphi)$ and declare that a user enters the near field regime once this loss exceeds a prescribed threshold $\delta_{\mathrm{gain}}$. Accordingly, the ERD boundary is defined as the largest distance that violates the far field approximation, i.e.,
\begin{equation}
r_{\mathrm{ERD}}(\rho,\varphi) \triangleq \sup\Bigl\{\,r>0:\;1-\mu(r,\rho,\varphi)\ge \delta_{\mathrm{gain}} \Bigr\}.
\label{eq:ERD_def}
\end{equation}

Notably, $\mu(r,\rho,\varphi)$ directly quantifies the gain degradation incurred when using a far field beam to serve a near field user. Hence, the ERD provides an operational and direction-dependent boundary for near field beam training, which is more informative than the classical RD for curved apertures. Based on the definition in \eqref{eq:ERD_def}, we obtain the following closed-form expression.
\begin{corollary}\label{cor:ERD}
For a given direction $(\rho,\varphi)$, the ERD of the 1-D CuRA is given by
\begin{equation}
r_{\mathrm{ERD}}(\rho,\varphi)
= \frac{\pi R^2\,\xi(\rho,\varphi,\beta)}{\lambda\,\big|J_0^{-1}(1-\delta_{\mathrm{gain}})\big|}.
\label{eq:ERD}
\end{equation}
\end{corollary}
\textit{Proof:} See Appendix \ref{app:proof_ERD}. \hfill$\blacksquare$\par 

It follows from \eqref{eq:ERD} that $r_{\mathrm{ERD}}$ scales quadratically with the curvature radius $R$ and linearly with the direction-dependent focusing parameter $\xi(\rho,\varphi,\beta)$, and is inversely proportional to the wavelength $\lambda$. This indicates that the near field focusing region extends to larger ranges along directions where the array presents a larger effective electrical aperture. Accordingly, Corollary~\ref{cor:ERD} provides a computable, direction-dependent criterion to partition the service region into a near field focusing regime ($r<r_{\mathrm{ERD}}$) and a far field steering regime ($r\ge r_{\mathrm{ERD}}$), which will be used to guide the subsequent codebook design.

Building on this criterion, we develop an ERD-guided hierarchical codebook. For each angular grid point $(\rho,\varphi)$ in the angular codebook, we adapt the distance sampling resolution according to the corresponding $r_{\mathrm{ERD}}(\rho,\varphi)$, thereby allocating more samples only where near field focusing is required.

First, for the near field focusing regime $r_{\min}\le r< r_{\mathrm{ERD}}(\rho,\varphi)$, where wavefront curvature is non-negligible, explicit distance sampling is required to maintain accurate focusing. We design the sampling resolution in the reciprocal-range domain $\tau\triangleq 1/r$, which naturally matches the spherical-wave phase structure. Given a prescribed range-domain correlation threshold $\delta_r$, the reciprocal-range spacing that ensures $G_r\le \delta_r$ between adjacent distance samples is $\Delta\tau
= \frac{2\lambda\, J_{0}^{-1}(\delta_r)}{\pi R^{2}\,\xi(\rho,\varphi,\beta)}$.
Accordingly, we sample $\tau$ over the interval $[\tau_{\min}, \tau_{\max}] = [1/r_{\mathrm{ERD}}(\rho,\varphi),\; 1/r_{\min}]$, and $\tau_{\max}>\tau_{\min}$ since $r_{\min}<r_{\mathrm{ERD}}(\rho,\varphi)$. A uniform grid in the $\tau$-domain is then constructed as
\begin{equation}
\tau_s=\tau_{\min}+(s-1)\Delta\tau,~ s=1,2,\ldots,M_{\mathrm{NF}},
\label{eq:tau_grid}
\end{equation}
where $M_{\mathrm{NF}}=\left\lfloor\frac{\tau_{\max}-\tau_{\min}}{\Delta\tau}\right\rfloor+1$.
Mapping back to the distance domain via $r_s=1/\tau_s$ yields non-uniform distance samples that are denser at shorter distances. For each $r_s$, we compute the corresponding near field beamfocusing vector $\mathbf b(r_s,\rho,\varphi)$.

Second, for the far field regime $r\ge r_{\mathrm{ERD}}(\rho,\varphi)$, the incident wavefront is well approximated as planar and the gain loss incurred by far field steering remains below the prescribed tolerance $\delta_{\mathrm{gain}}$. Nevertheless, to accommodate users close to the boundary and to ensure a smooth transition, we introduce a sparse set of additional distance samples using a logarithmic grid over $[r_{\mathrm{ERD}}(\rho,\varphi),\, r_{\max}]$:
\begin{equation}
r_l
= r_{\mathrm{ERD}}(\rho,\varphi)\left(\frac{r_{\max}}{r_{\mathrm{ERD}}(\rho,\varphi)}\right)^{\frac{l-1}{M_{\mathrm{FF}}-1}},
~ l=1,2,\ldots,M_{\mathrm{FF}},
\label{eq:farfield_sampling}
\end{equation}
where $M_{\mathrm{FF}}=\max\!\left(2,\left\lceil \frac{\log\!\big(r_{\max}/r_{\mathrm{ERD}}(\rho,\varphi)\big)}{\log(1+\eta_r)} \right\rceil\right)$ and $\eta_r>0$ controls the sparsity of the far field sampling. For each $r_l$, we compute the corresponding beamfocusing vector $\mathbf b(r_l,\rho,\varphi)$. In practice, these vectors become nearly identical when $r\gg r_{\mathrm{ERD}}(\rho,\varphi)$. 

The proposed framework concentrates sampling and design in the near field focusing regime, where range dependence is essential. By sampling uniformly in the reciprocal-range domain $\tau=1/r$, it naturally induces a distance-adaptive grid in $r$ that becomes denser at short ranges where wavefront curvature is most pronounced, and progressively sparser as $r$ increases. In the far field, only a small number of range samples are retained since the array response becomes nearly range-invariant. As a result, the proposed codebook captures the spherical-wave characteristics where necessary while keeping the overall size compact, thereby reducing beam-training overhead. Based on the above analysis, Algorithm~\ref{alg:polar-domain-arched} summarizes the proposed hierarchical codebook construction, which combines polar-domain angular sampling with ERD-guided range sampling for the 1-D CuRA.

\begin{algorithm}
\caption{ERD-Guided Hierarchical Polar-Domain Codebook Generation for the 1-D CuRA}
\label{alg:polar-domain-arched}
\begin{algorithmic}[1]
\Require Array: $R,\beta,N,\lambda$; coverage: $\rho_{\max}, r_{\min}, r_{\max}$;
thresholds: $\delta_p,\delta_r,\delta_{\mathrm{gain}}$; sparsity: $\eta_r$.
\Ensure Hierarchical codebook $\mathcal W_{\mathrm{total}}$.

\State Compute $\Delta\rho$ via \eqref{eq:delta_rho}. 
\State $I_{\max}\gets \left\lfloor \rho_{\max}/\Delta\rho \right\rfloor$,~$\mathcal G_{\mathrm{angle}} \gets \{(\rho_0=0,\varphi_0=0)\}$.
\For{$i=1$ to $I_{\max}$}
    \State $\rho_i \gets i\,\Delta\rho$.
    \State Compute $N_{\varphi}(\rho_i)$ via \eqref{eq:num_angular_samples}.
    \For{$j=0$ to $N_{\varphi}(\rho_i)-1$}
        \State $\varphi_{ij} \gets -\beta + j\cdot \frac{2\beta}{N_{\varphi}(\rho_i)}$.
        \State $\mathcal G_{\mathrm{angle}} \gets \mathcal G_{\mathrm{angle}} \cup \{(\rho_i,\varphi_{ij})\}$.
    \EndFor
\EndFor

\State $\mathcal W_{\mathrm{total}} \gets \emptyset$.
\For{each $(\rho_i,\varphi_{ij})\in \mathcal G_{\mathrm{angle}}$}

    \State $r_{\mathrm{ERD}}^{ij}\gets r_{\mathrm{ERD}}(\rho_i,\varphi_{ij})$ via \eqref{eq:ERD}.

    \If {$r_{\min} \leq r < r_{\text{ERD}}^{ij}$} generate $\{r_s\}_{s=0}^{M_L-1}$ via \eqref{eq:tau_grid}.

    \State $\mathcal{W}_{\text{total}} \leftarrow \mathcal{W}_{\text{total}} \cup \{ \mathbf{b}(r_s, \rho_i, \varphi_{ij})\}$.
 \EndIf   
    \If {$r_{\max} > r_{\text{ERD}}^{ij}$} generate $\{r_l\}_{l=1}^{M_L}$ via \eqref{eq:farfield_sampling}.
    \State $\mathcal{W}_{\text{total}} \leftarrow \mathcal{W}_{\text{total}} \cup \{\mathbf{b}(r_l, \rho_i, \varphi_{ij})\}$.
    \EndIf
\EndFor
\State \Return $\mathcal W_{\mathrm{total}}$
\end{algorithmic}
\end{algorithm}

\section{Generalization to the 2-D CuRA}\label{sec:arcUPA}
With the evolution of near field wireless systems toward higher frequencies and denser deployments, curvature-reconfigurable apertures provide a flexible way to improve beamfocusing and wide-angle coverage. Building on the 1-D CuRA analysis, we next consider a 2-D CuRA, which is constructed by stacking multiple 1-D CuRA branches along an orthogonal linear dimension. Following the methodology developed for the 1-D CuRA, we first describe the geometric configuration of the 2-D CuRA and then develop its corresponding hierarchical codebook design.

\subsection{Geometry of the 2-D CuRA}
The geometry of the proposed 2-D CuRA is illustrated in Fig.~\ref{fig:sysUPA}. It consists of $M$ identical 1-D CuRAs placed in parallel and indexed along the $x$-axis with uniform inter-ULA spacing $d_x$. Each 1-D CuRA has a common curvature radius $R$ and bending angle $\beta$, and comprises $N$ antenna elements. Accordingly, the position of the $(m,n)$-th element is 
$\mathbf q_{m,n}
=\Big(m d_x,\;R\big(\cos(\beta-\psi_n)-\cos\beta\big),\;R\sin(\beta-\psi_n)\Big)$, where $m=1,\ldots,M$, $n=1,2,\ldots,N$, and $\psi_n={2\beta(n-1)}/{(N-1)}$ denotes the arc parameter within each 1-D CuRA.

Thus, the near field beamfocusing vector of the 2-D CuRA is obtained by compensating the spherical-wave phase across the 2-D aperture. Letting $r_{m,n}=\big\|\mathbf u(r,\theta,\phi)-\mathbf q_{m,n}\big\|_2$ denote the propagation distance to the $(m,n)$-th element, the vector is given by
\begin{equation}
\begin{aligned}
\mathbf b(r,\theta,\phi)
=\frac{1}{\sqrt{MN}}
\Big[\,&e^{-j\frac{2\pi}{\lambda}(r_{1,1}-r)},\ \ldots,\ e^{-j\frac{2\pi}{\lambda}(r_{M,N}-r)}\Big]^{\mathrm T},
\label{eq:b_upa}
\end{aligned}
\end{equation}
where subtracting $r$ removes the common phase and retains only the relative phase variations across the aperture. 
The phase of each element is thus determined by $r_{m,n}$ can be approximated via a Taylor expansion in $ 1/r$ as
\begin{equation}
    \begin{aligned}
       & r_{m,n} \approx r - \Phi^{(m,n)}_{\text{f}}(\theta,\phi)+ \frac{1}{2r}\Bigl[ m^2 d_x^2 +R^2\cos\beta\\\
        &\times\bigl(1 - 2\cos(\beta-\psi_n) + \cos^2\beta\bigr) - \bigl(\Phi^{(m,n)}_{\text{f}}(\theta,\phi))^2 \Bigr], \label{dis-UPA}
    \end{aligned}
\end{equation}
where $\Phi^{(m,n)}_{\text{f}}(\theta,\phi)=[ m d_x \sin\theta\cos\phi + R\sin\theta\sin\phi(\cos(\beta-\psi_n)-\cos\beta)  + R\cos\theta\sin(\beta-\psi_n)]$ is the first-order term.
\begin{figure}[!t]
    \centering 
   \includegraphics[width=0.45\textwidth]{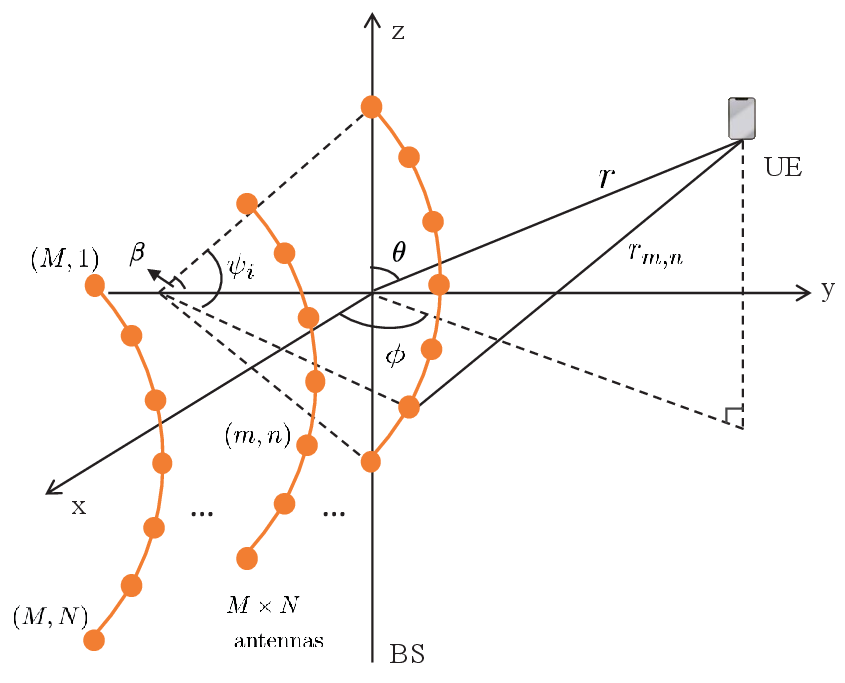}
    \caption{Geometry of the 2-D CuRA.} 
    \label{fig:sysUPA} 
        \vspace{-1.5 em}
\end{figure}
It is worth noting that, when the link distance satisfies $r \gg R$, the spherical-wave near-field response reduces to the conventional far-field steering vector, whose elements are given by
$\big[\mathbf a(\theta,\phi)\big]_{m,n}
= \frac{1}{\sqrt{MN}}\exp\!\left(j\frac{2\pi}{\lambda}\, \Phi^{(m,n)}_{\mathrm f} (\theta,\phi)\right)$.

\subsection{Codebook Design for the 2-D CuRA}

Following the 1-D CuRA design, we construct a dedicated codebook for the 2-D CuRA based on the beamforming-gain criterion. We adopt a hierarchical structure, an angular codebook is designed to provide directional coverage, and for each angular direction, the range codewords are generated via non-uniform sampling to capture the distance-dependent near field behavior.

For the angular parametrization, we reuse the polar-domain representation introduced in Section~\ref{sec:polar} to describe the departure direction with respect to the array normal (the $x$-axis) and its rotation in the orthogonal plane. Specifically, we define $\rho = \arccos(\sin\theta\cos\phi)$ and $\varphi = \arctan(\frac{\sin\theta\sin\phi}{\cos\theta})$. Compared with the 1-D CuRA, the 2-D CuRA additionally introduces a linear aperture along the $x$-dimension, which can be handled by applying the same hierarchical procedure while incorporating the $x$-dimension sampling rules derived from the 2-D CuRA geometry in \eqref{dis-UPA}.

\subsubsection{Joint angular-domain correlation}

For a fixed link distance $r$, consider two directions $\mathbf s_1=(r,\rho_1,\varphi_1)$ and $\mathbf s_2=(r,\rho_2,\varphi_2)$. Owing to the separable aperture structure of the 2-D CuRA, the angular-domain beamforming gain can be approximated by a product form
\begin{equation}
G_{\text{angle}}^{\mathrm{P}}(\mathbf s_1,\mathbf s_2)
= G_{\mathrm v}(\Delta\rho,\Delta\varphi)\, G_{\mathrm h}(\Delta u),\label{G_angle}
\end{equation}
where $G_{\mathrm v}(\Delta\rho,\Delta\varphi)$ captures the contribution from the curved dimension (the $yz$-plane) and is given by \eqref{eq:G_a}, and $G_{\mathrm h}(\Delta u)$ captures the contribution from the linear dimension (the $x$-axis). In particular, letting $\Delta u$ denote the difference in the direction cosine along the $x$-axis, $\Delta u = \sin\rho_2\cos\varphi_2-\sin\rho_1\cos\varphi_1$, the horizontal contribution follows the classical ULA array-factor form 
as $G_{\text{h }}(\Delta u) = \left| \frac{\sin\left(\frac{ \pi  M}{\lambda}   d_x \Delta u\right)}{M \sin\left(\frac{ \pi  }{\lambda} d_x \Delta u\right)} \right|$ \cite{ula}. To control redundancy, we enforce a joint angular-domain coherence constraint $G_{\text{angle}}^{\mathrm P}(\mathbf s_1,\mathbf s_2)\le \eta_a$ for adjacent grid points. Since \eqref{G_angle}, a sufficient approach is to design the vertical and horizontal grids separately using tightened per-dimension thresholds. This separable design guarantees the joint constraint while retaining low-complexity grid construction.

First, we derive the radial sampling interval by considering purely radial perturbations on the polar grid, i.e., $\Delta\varphi=0$. Under this condition, the joint constraint $G_{\text{angle}}^{\mathrm P}\le \eta_a$ can be enforced via the separable tightened thresholds $G_{\mathrm v}\le \eta_a^{1/2}$ and $G_{\mathrm h}\le \eta_a^{1/2}$, which yields the two independent inequalities
\begin{equation}
\begin{aligned}
& \left| J_0\!\left( \frac{2\pi R}{\lambda}\,\Theta(\rho,\varphi)\,|\Delta\rho|\,\frac{\sin\beta}{\beta} \right) \right|
\leq \eta_a^{1/2}, G_{\mathrm h}(\Delta u) \le \eta_a^{1/2},
\end{aligned}
\label{eq:separate_radial_constraints}
\end{equation}
where $\Delta u$ denotes the direction-cosine difference along the $x$-axis. For $\Delta\varphi=0$, we have $\Delta u \approx |\cos\rho\,\cos\varphi|\,|\Delta\rho|$, and then we have $|\cos\rho\,\cos\varphi|\,|\Delta\rho|
\le \frac{\lambda}{M d_x}\cdot \frac{\arcsin(\eta_a^{1/2})}{\pi}$.
Following a similar process in 1-D CURA, solving the vertical constraint in \eqref{eq:separate_radial_constraints} gives $\Delta\rho_{\mathrm v}
= \frac{J_0^{-1}(\eta_a^{1/2})\,\lambda}{2\pi R\,\Theta(\rho,\varphi)}\cdot \frac{\beta}{\sin\beta}$. Likewise, the horizontal constraint yields $\Delta\rho_{\mathrm h}
= \frac{\lambda}{M d_x\,|\cos\rho\,\cos\varphi|}\cdot \frac{\arcsin(\eta_a^{1/2})}{\pi}$. Since both constraints must hold, we set the effective radial step size as
\begin{equation}
\Delta\rho^{2\text{D}}=\min\{\Delta\rho_{\mathrm v},\,\Delta\rho_{\mathrm h}\}.\label{49}
\end{equation}
The radial sampling points are then generated by $\rho_i=i\,\Delta\rho^{2\text{D}}$, $ i=0,1,\ldots,I_{\max}$ with $I_{\max}=\left\lfloor \rho_{\max}/\Delta\rho^{2\text{D}}\right \rfloor$.

Second, for circumferential perturbations on a fixed radial layer (\(\Delta\rho=0\)), the polar-domain distance is locally approximated as \(d_p \approx \rho|\Delta\varphi|\). Under this condition, the joint constraint separates into \(G_{\mathrm v}\le \eta_a^{1/2}\) and \(G_{\mathrm h}\le \eta_a^{1/2}\) yielding
\begin{equation}
\begin{aligned}
&\left| J_0\!\left( \frac{2\pi R}{\lambda}\,\Theta(\rho,\varphi)\,\rho|\Delta\varphi|\,\frac{\sin\beta}{\beta} \right) \right|
\le \eta_a^{1/2},\\
&|\cos\rho\,\cos\varphi|\,|\Delta\varphi|
\le \frac{\lambda}{M d_x}\cdot \frac{\arcsin(\eta_a^{1/2})}{\pi},
\end{aligned}
\label{eq:separate}
\end{equation}
where the factor $\rho$ in the first inequality accounts for the arc-length mapping from an azimuthal perturbation $\Delta\varphi$ to the polar-domain distance $\rho\Delta\varphi$ on the ring at radius $\rho$.
Solving \eqref{eq:separate} yields the circumferential step sizes $\Delta\varphi_{\mathrm v}
= \frac{J_0^{-1}(\eta_a^{1/2})\,\lambda}{2\pi R\,\Theta(\rho,\varphi)\,\rho}\cdot \frac{\beta}{\sin\beta}$ and $\Delta\varphi_{\mathrm h}
= \frac{\lambda}{M d_x\,|\cos\rho\,\cos\varphi|}\cdot \frac{\arcsin(\eta_a^{1/2})}{\pi}$.
Analogous to the radial case, the effective circumferential spacing is chosen as
\begin{equation}
\Delta\varphi^{2\text{D}}(\rho,\varphi)=\min\{\Delta\varphi_{\mathrm v},\,\Delta\varphi_{\mathrm h}\}.\label{53}
\end{equation}
For a limited azimuthal search region $\varphi\in[-\beta,\beta]$, the number of circumferential samples on the $i$-th ring ($\rho_i>0$) is then
\begin{equation}
N_{\varphi}(\rho_i)=\left\lceil \frac{2\beta}{\Delta\varphi^{2\text{D}}(\rho_i,0)}\right\rceil,
\label{eq:upa_layer_correct}
\end{equation}
The azimuth samples are uniformly placed as
\begin{equation}
\varphi_{i,j}=-\beta + j\cdot \frac{2\beta}{N_{\varphi}(\rho_i)},\quad
j=0,1,\ldots,N_{\varphi}(\rho_i)-1.
\label{eq:upa_coordinates_correct}
\end{equation}
The total number of angular grid points is
\begin{equation}
M_{\text{angle}} = 1 + {\textstyle \sum_{i=1}^{I_{\max}}} N_{\varphi}(\rho_i).
\label{eq:upa_total_angular_samples_correct}
\end{equation}
By construction, adjacent samples satisfy the joint coherence constraint and account for the distinct angular resolutions of the curved ($yz$) and linear ($x$) aperture dimensions. The resulting polar-domain grid provides efficient angular coverage with limited redundancy.
For an angular grid point $(\rho_i,\varphi_{i,j})$, the far-field angular codeword is constructed via a separable Kronecker structure. The steering vector for the $x$-axis is
\begin{equation}
\begin{aligned}
\big[\mathbf w_x(\rho,\varphi)\big]_m
= \frac{1}{\sqrt{M}}
\exp\!\left(j\frac{2\pi}{\lambda} m d_x \sin\rho\cos\varphi\right). 
\label{eq:wx_def}    
\end{aligned}
\end{equation}
The steering vector for the curved $yz$-dimension is
\begin{equation}
\big[\mathbf w_{yz}(\rho,\varphi)\big]_n
= \frac{1}{\sqrt{N}}
\exp\!\left(j\,\Phi^{(n)}_{\mathrm f}(\rho,\varphi)\right), n=1,2,\ldots,N,
\label{eq:wyz_def}
\end{equation}
where $\Phi^{(n)}_{\mathrm f}(\rho,\varphi)$ is from \eqref{dis-UPA} for the 1-D CuRA cross-section. Thus, the angular-domain codeword for the 2-D CuRA is
\begin{equation}
\mathbf w_{\text{angle}}(\rho_i,\varphi_{i,j})
= \mathbf w_x(\rho_i,\varphi_{i,j}) \otimes \mathbf w_{yz}(\rho_i,\varphi_{i,j}).
\label{eq:angular_codeword}
\end{equation}

\subsubsection{Distance Domain}

For the 2-D CuRA, the distance-domain correlation must capture the combined impact of the curved aperture and the linear aperture. Consider two user locations $\mathbf s_1=(r_1,\rho,\varphi)$ and $\mathbf s_2=(r_2,\rho,\varphi)$
that share the same angular direction but differ in range. The corresponding beamforming gain is
\begin{equation}
\begin{aligned}
G_r^{\mathrm P}(\mathbf s_1,\mathbf s_2)
&= \big|\mathbf b^{\mathrm H}(r_1,\rho,\varphi)\,\mathbf b(r_2,\rho,\varphi)\big|\\
&= \left|\frac{1}{MN}\sum_{m=1}^{M}\sum_{n=1}^{N}
e^{-j \frac{2\pi}{\lambda}\left(r_{m,n}^{(2)}-r_{m,n}^{(1)}\right)}\right|.    
\label{eq:exact_distance_gain}
\end{aligned}
\end{equation} 
Exploiting the separable structure of the 2-D CuRA aperture, the double sum in \eqref{eq:exact_distance_gain} can be approximated by a product of a vertical term and a horizontal term, i.e.,
\begin{equation}
G_r^{\mathrm P}(\mathbf s_1,\mathbf s_2)
\approx G_{r,\mathrm v}^{\mathrm P}(\mathbf s_1,\mathbf s_2)\,
G_{r,\mathrm h}^{\mathrm P}(\mathbf s_1,\mathbf s_2),
\label{eq:separated_gain}
\end{equation}
where $G_{r,\mathrm v}^{\mathrm P}(\mathbf s_1,\mathbf s_2)$ corresponds to the curved $yz$-plane and follows the 1-D CuRA distance domain approximation in \eqref{eq:Gr}. 

For the horizontal $x$-axis, a quadratic-phase approximation in the reciprocal-distance domain leads to a Fresnel-integral form \cite{ula}:
\begin{equation}
G_{r,\mathrm h}^{\mathrm P}(\mathbf s_1,\mathbf s_2)
\approx \left|\frac{C\!\left(\sqrt{\epsilon}\right)+j\,S\!\left(\sqrt{\epsilon}\right)}{\sqrt{\epsilon}}\right| \triangleq G(\epsilon),
\label{eq:Gh}
\end{equation}
where $C(\cdot)$ and $S(\cdot)$ are the Fresnel integrals
$C(\eta)=\int_{0}^{\eta}\cos\!\left(\frac{\pi}{2}t^{2}\right)\,dt$
and
$S(\eta)=\int_{0}^{\eta}\sin\!\left(\frac{\pi}{2}t^{2}\right)\,dt$.
The parameter $\epsilon$ is proportional to the reciprocal-range separation $\Delta\tau \triangleq \left|\frac{1}{r_2}-\frac{1}{r_1}\right|$, and is given by $\epsilon = \frac{d_x^2 M^2}{\lambda}\,\Delta\tau$. It can be observed that the $x$-axis of the 2-D CuRA forms a conventional ULA with aperture length $L_x=(M-1)d_x$, whereas the $yz$-plane forms an arched array with arc length $L_{\mathrm{arc}}=2R\beta$. Consequently, a single classical RD is not directly applicable to the 2-D CuRA. The effective aperture is anisotropic across the two orthogonal axes, and the focusing behavior depends on the spatial direction $(\rho,\varphi)$. 

To capture this anisotropy, we adopt the ERD as a direction dependent boundary for significant near field effects. For a given direction $(\rho,\varphi)$, define the normalized correlation between the near field response and its far field counterpart as $\mu^{\mathrm P}(r,\rho,\varphi)\triangleq \big|\mathbf b^{\mathrm H}(r,\rho,\varphi)\,\mathbf a(\rho,\varphi)\big|$, With a prescribed gain-loss tolerance $\delta_{\mathrm{gain}}^{\mathrm P}$, a link is regarded as near field once $1-\mu^{\mathrm P}(r,\rho,\varphi)\ge \delta_{\mathrm{gain}}^{\mathrm P}$.
Since the 2-D CuRA exhibits separable apertures along the $x$-axis and the arched ($yz$) dimension, we approximate the ERD boundary by the more stringent of the two dimension-wise ERDs. From the $x$-axis, which exhibits a $\operatorname{sinc}$ type beam pattern, the ERD is given by $ r_{\mathrm{ERD}}^{x}(\rho,\varphi) = \frac{\pi L_x^2 \Xi_x(\rho,\varphi)}{4 \lambda \bigl| J_0^{-1}(1-\delta_{\mathrm{gain}}) \bigr|}$, with $\Xi_x(\rho,\varphi) = \sin^2\rho\cos^2\varphi$ is the $x$-axis geometric projection factor that accounts for the reduced effective aperture when observing from direction $(\rho,\varphi)$. From the arc-direction, the solution has been provided in Corollary \ref{cor:ERD}, with the specific expression given in \eqref{eq:ERD}. Thus, we obtain the compact form as  
\begin{equation}
\begin{aligned}
     & r_{\mathrm{ERD}}^{2\text{D}}(\rho,\varphi) \\
      &\quad \approx \min\Biggl( 
    \frac{\pi L_x^2 \Xi_x(\rho,\varphi)}{4\lambda \bigl| J_0^{-1}(1-\delta_{\mathrm{gain}}) \bigr|}, \;
    \frac{\pi R^2 \xi(\rho,\varphi,\beta)}{\lambda \bigl| J_0^{-1}(1-\delta_{\mathrm{gain}}) \bigr|}
    \Biggr).
    \label{eq:ERD_UPA_compact}
\end{aligned}
\end{equation} 

Based on the above analysis, we partition the distance domain into a near field regime, where spherical-wave focusing induces strong range dependence and thus requires dense sampling, and a far field regime, where the response is nearly range-invariant and sparse sampling suffices. To determine the sampling interval between 2-D CuRA codewords, we adopt a design approach analogous to that for ULAs and introduce the reciprocal range $\tau\triangleq 1/r$ as the sampling variable. We fix a prescribed distance-domain correlation threshold $\delta_r^{\mathrm P}$.

For horizontal direction ($x$-axis), to suppress redundancy in the horizontal dimension, we enforce that the distance-domain correlation in \eqref{eq:Gh} does not exceed a prescribed threshold $\delta_r^{\mathrm P}$, i.e., $G_{r,h}^{\mathrm P}\le \delta_r^{\mathrm P}$. Let $\epsilon_\Delta$ denote the solution to $G(\epsilon_\Delta)=\delta_r^{\mathrm P}$. Since $G(\epsilon)$ decreases with $\epsilon$ over the operating range, the constraint $G(\epsilon)\le \delta_r^{\mathrm P}$ is satisfied whenever $\epsilon\ge \epsilon_\Delta$. Thus, we obtain the required reciprocal-range separation along the $x$-axis as $|\Delta\tau_h|= \frac{\lambda\,\epsilon_\Delta}{R d_x^2 M^2}$. In summary, to ensure that adjacent distance samples maintain a correlation no smaller than the prescribed threshold $\delta_r^{\mathrm P}$ in both aperture dimensions, the reciprocal-range step size must satisfy
\begin{equation}
\Delta\tau \leq \min\left(
\frac{\lambda\,\epsilon_\Delta}{R d_x^2 M^2},\ 
\frac{2\lambda\,J_0^{-1}(\delta_r^{\mathrm P})}{\pi R^2\,\xi(\rho,\varphi,\beta)}
\right).
\label{eq:dis_tau}
\end{equation}
Uniform sampling in the $\tau$-domain over $[\tau_{\min},\tau_{\max}]$ gives
\begin{equation}
\tau_z = \tau_{\min} + z\,\Delta\tau,\quad z=0,1,\ldots,Z,
\label{eq:sam}
\end{equation}
where $Z=\left\lfloor\frac{\tau_{\max}-\tau_{\min}}{\Delta\tau}\right\rfloor$.
Mapping back to the distance domain yields a non-uniform grid
\begin{equation}
r_z = \frac{1}{\tau_z}
= \frac{1}{\tau_{\min}+z\,\Delta\tau},\label{72}
\end{equation}
which is denser at shorter distances. Here, $r_{\min}$ denotes the minimum communication distance of interest, satisfying $r_{\min} > 0.5\sqrt{D^3/\lambda}$, where $D$ represents the array aperture. 

The near-field codeword for a target distance $r$ at a given polar angle pair $(\rho_i, \varphi_{i,j})$ is obtained by applying a distance dependent phase correction to the angular domain base vector, expressed as
\begin{equation}
    \mathbf{w}(r, \rho_i, \varphi_{i,j}) = \mathbf{w}_{\text{angle}}(\rho_i,\varphi_{i,j}) \odot \mathbf{p}(r; \rho_i,\varphi_{i,j}),\label{codewords}
\end{equation}
where $\odot$ denotes the Hadamard product, and $\mathbf{p}(r; \rho_i, \varphi_{i,j})$ is the phase correction vector accounting for the spherical wavefront curvature. The phase of its $(m,n)$-th element, representing the additional compensation relative to a reference distance $r_{\text{ref}}$, is given by
\begin{equation}
    [\mathbf{p}(r; \rho_i,\varphi_{i,j})]_{m,n} = \exp\left( -j \frac{2 \pi}{\lambda} \left( r_{m,n}(r) - r_{m,n}(r_{\text{ref}}) \right) \right).\label{core}
\end{equation}

Built upon this design, the computational complexity and storage requirements of the 3D codebook are greatly reduced. Thus, the overall codebook design process integrating these sampling methods is summarized in Algorithm \ref{alg:upa}.

\begin{algorithm}[!t]
\caption{ERD-Guided Hierarchical Polar-Domain Codebook Generation for the 2-D CuRA}
\label{alg:upa}
\begin{algorithmic}[1]
\Require Array parameters: $R,\beta,N,M,d_x,\lambda$; 
        coverage: $\rho_{\max}, r_{\min}, r_{\max}$;
        thresholds: $\eta_a,\delta_r,\delta_{\mathrm{gain}}$.
\Ensure Hierarchical polar codebook $\mathcal W_{\mathrm{total}}$.

\State Compute $\Delta\rho^{2\text{D}}$ via \eqref{49}. 
\State $I_{\max}\gets \left\lfloor \rho_{\max}/\Delta\rho^{2\text{D}} \right\rfloor$, 
    $\mathcal G_{\mathrm{angle}} \gets \{(\rho_0=0,\varphi_0=0)\}$.
\For{$i=1$ to $I_{\max}$}
    \State $\rho_i \gets i\,\Delta\rho^{2\text{D}}$.
    \State Compute $\Delta\varphi^{2\text{D}}(\rho_i,0)$ via \eqref{53}, $N_{\varphi}(\rho_i)$ via \eqref{eq:upa_layer_correct}.
    \For{$j=0$ to $N_{\varphi}(\rho_i)-1$}
        \State $\varphi_{ij}$ via \eqref{eq:upa_coordinates_correct}, $\mathcal G_{\mathrm{angle}} \gets \mathcal G_{\mathrm{angle}} \cup \{(\rho_i,\varphi_{ij})\}$.
        \For{each $(\rho_i,\varphi_{ij})\in \mathcal G_{\mathrm{angle}}$}
            \State  $\xi_{ij} \gets \xi(\rho_i,\varphi_{ij},\beta)$ ,        $\Xi_x^{ij} \gets \Xi_x(\rho_i,\varphi_{ij})$.
          \State Generate $\mathbf{w}_{\mathrm{angle}}^{ij}$ via \eqref{eq:wx_def}, \eqref{eq:wyz_def}, and \eqref{eq:angular_codeword}.
          \EndFor
    \EndFor
\EndFor

\State $\mathcal W_{\mathrm{total}} \gets \emptyset$.
\For{each $(\rho_i,\varphi_{ij})\in \mathcal G_{\mathrm{angle}}$}
    \State Compute ERD $r_{\mathrm{ERD}}^{2\text{D}}(\rho,\varphi)$ via \eqref{eq:ERD_UPA_compact}. $\Delta\tau$ via \eqref{eq:dis_tau}. 
    \State Generate $r_z$ via \eqref{72}, $ [\mathbf{p}(r_z; \rho_i,\varphi_{i,j})]_{m,n}$ via \eqref{core}.  
    \State $\mathcal{W}_{\text{total}} \leftarrow \mathcal{W}_{\text{total}} \cup \{\mathbf{w}(r_z, \rho_i, \varphi_{ij})\}$ via \eqref{codewords}.
    \EndFor
    \State \Return $\mathcal W_{\mathrm{total}}$

\end{algorithmic}
\end{algorithm}

\section{Simulation Results} \label{sec:simulation}

In this section, we present simulation results to validate the theoretical analysis, assess the performance of the proposed hierarchical codebook, and demonstrate the near field beam-focusing characteristics of bendable arrays. Unless otherwise specified, the following parameters are used. The carrier frequency is $f_c=30~\mathrm{GHz}$, and the array comprises $N=512$ antenna elements with half-wavelength inter-element spacing. The bendable array is characterized by a bending angle $\beta\in[0,\pi/2]$, with $\theta\in[0,\pi]$ and $\phi\in[0,\pi]$. The correlation thresholds are set as $\delta_\theta=\delta_\phi=\delta_r=0.5$. To ensure a fair comparison, the proposed arched array is configured to have the same physical aperture as the benchmark XL-ULA in \cite{enable}.

\begin{figure}[th]
    \centering      
    \begin{subfigure}{0.24\textwidth}  
        \centering
        \includegraphics[width=\textwidth]{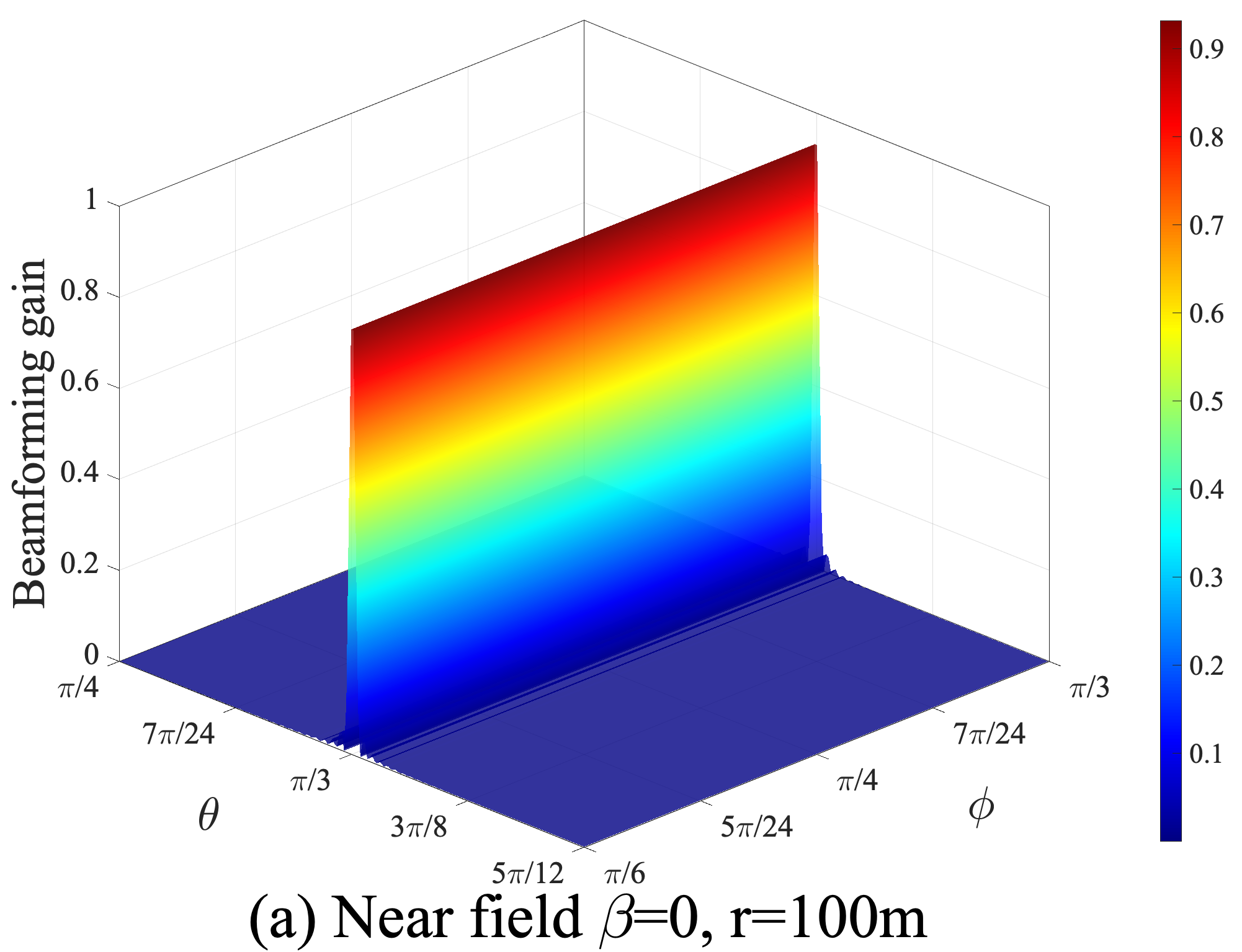}   
    \end{subfigure}    
    \hfill 
    \begin{subfigure}{0.24\textwidth}  
        \centering
        \includegraphics[width=\textwidth]{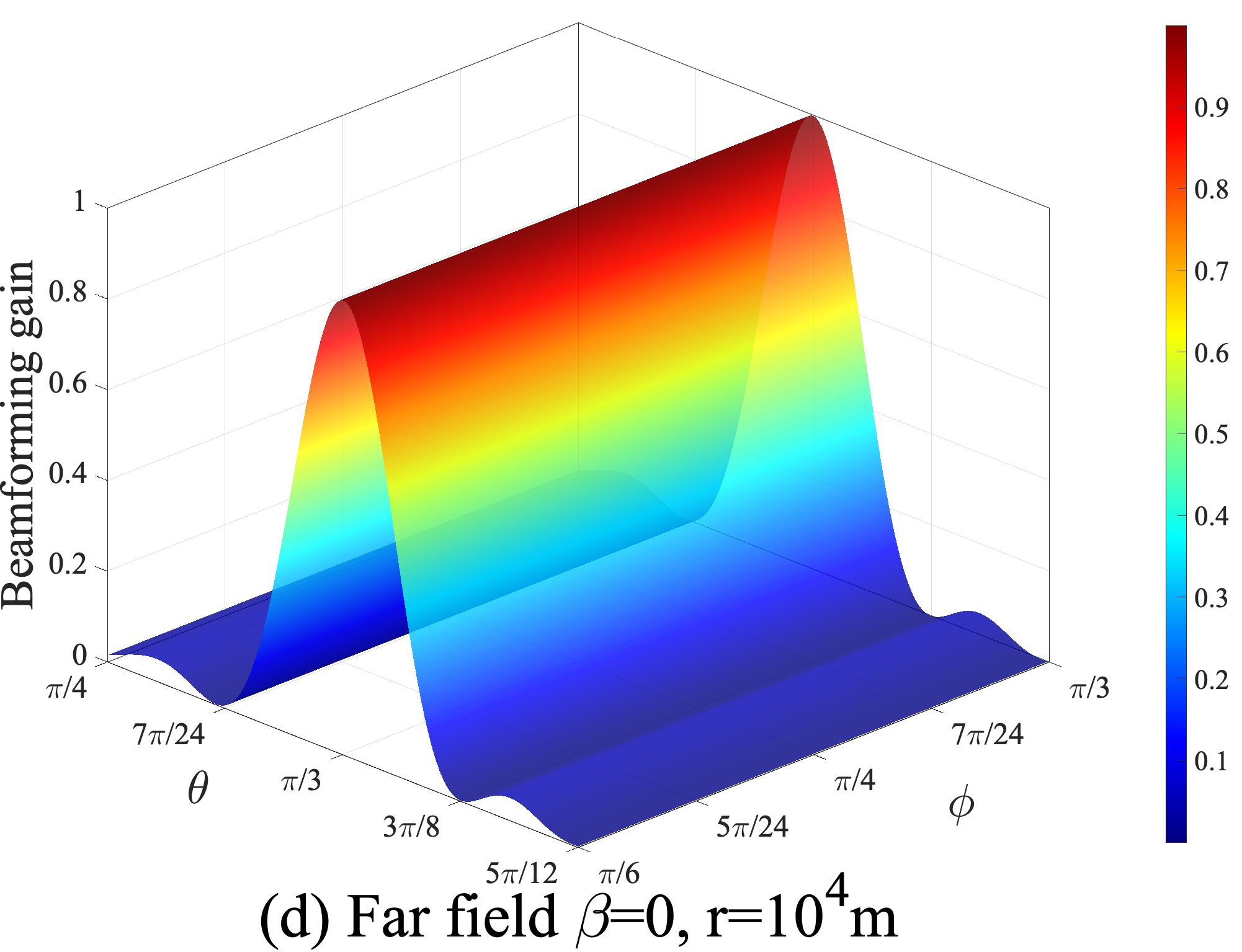}   
    \end{subfigure}    
    \hfill 
    \begin{subfigure}{0.24\textwidth} 
        \centering
        \includegraphics[width=\textwidth]{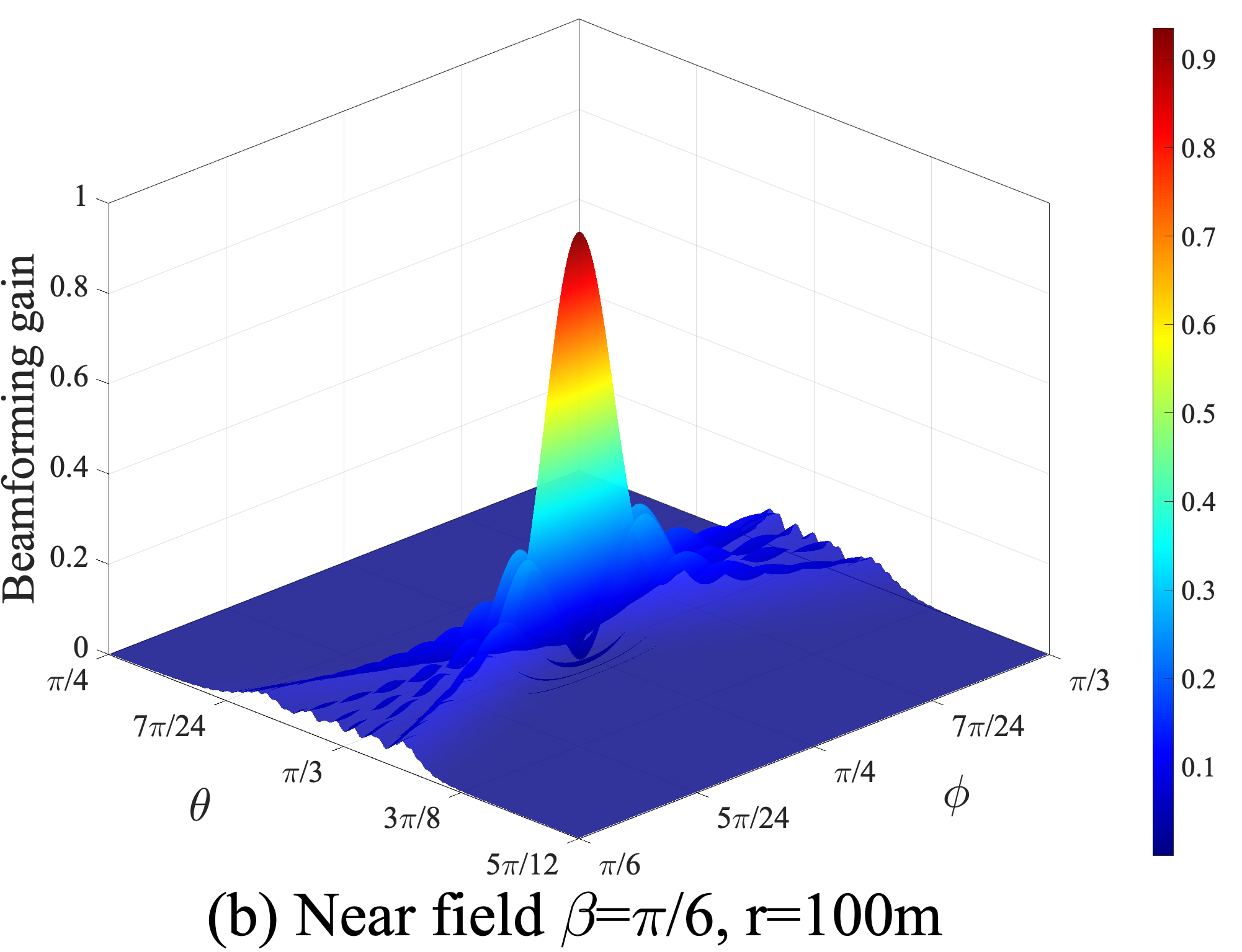}   
    \end{subfigure}
  \centering      
    \begin{subfigure}{0.24\textwidth}  
        \centering
        \includegraphics[width=\textwidth]{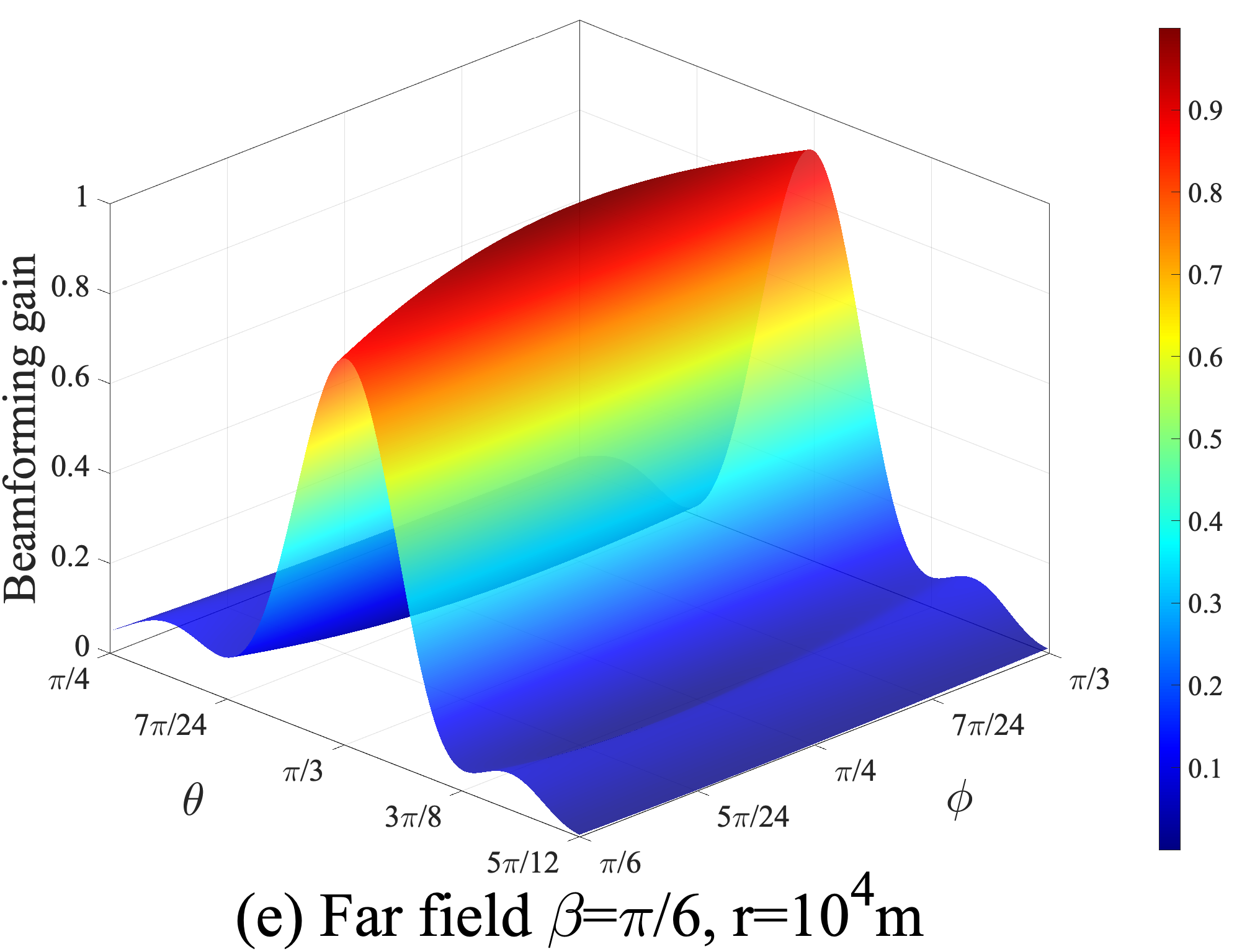}   
    \end{subfigure}    
    \hfill 
    \begin{subfigure}{0.24\textwidth}  
        \centering
        \includegraphics[width=\textwidth]{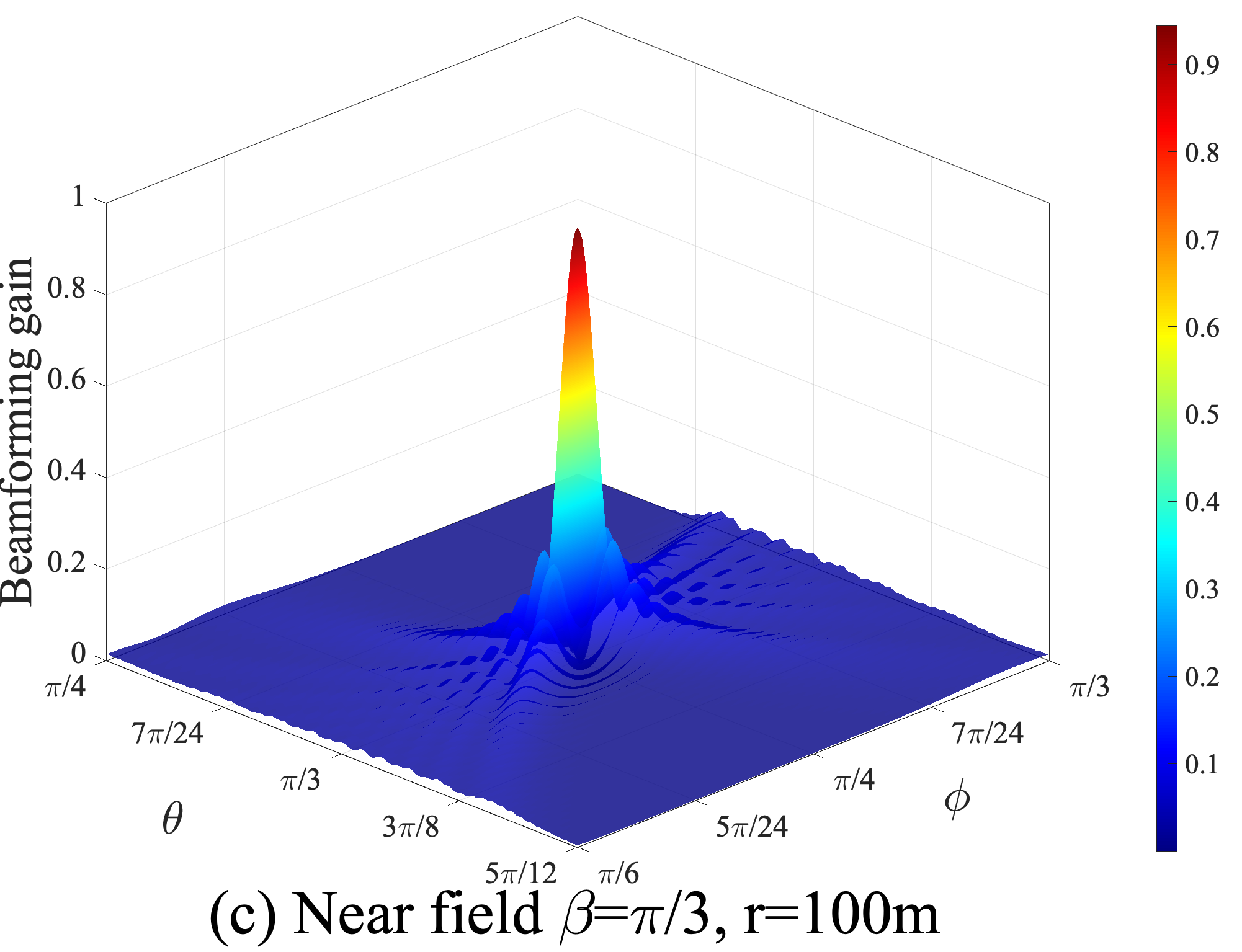}   
    \end{subfigure}    
    \hfill 
    \begin{subfigure}{0.24\textwidth} 
        \centering
        \includegraphics [width=\textwidth]{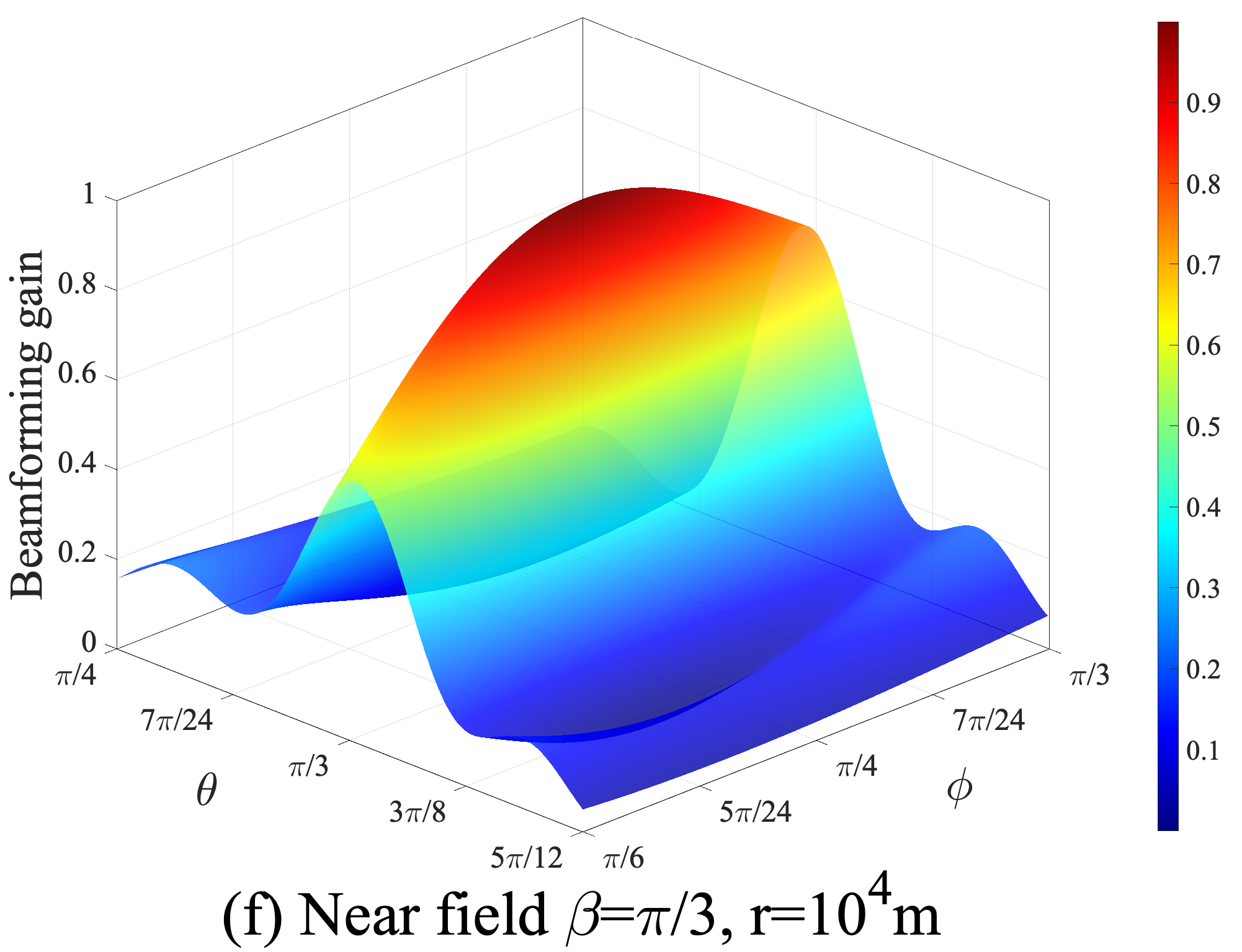}   
    \end{subfigure}  
    \caption{Beamforming gain in angular domain $(\theta-\phi)$ vs. distance and bending angle $\beta$ for the 1-D CuRA. }    
    \label{fig:beamforming_all} 
        \vspace{-1.5 em}
\end{figure}

We start by examining the beamforming-gain expression derived for the 1-D CuRA. Figs.~\ref{fig:beamforming_all}(a)--~\ref{fig:beamforming_all}(c) plot the near field beamforming gain at $r=100$ m for $\beta\in\{0, \pi/6, \pi/3\}$, while Fig.~\ref{fig:beamforming_all}(d)--Fig.~\ref{fig:beamforming_all}(f) show the corresponding far field patterns at $r=10^{4}$ m. For $\beta\rightarrow 0$, the 1-D CuRA reduces to a conventional ULA. As shown in Fig.~\ref{fig:beamforming_all}(a) and Fig.~\ref{fig:beamforming_all}(d), the angular response is essentially one-dimensional. The gain varies with $\theta$ but remains nearly invariant with respect to $\phi$, which leads to a ridge-shaped pattern over the $(\theta,\phi)$ domain. Moreover, the far field pattern exhibits the expected sinc-like mainlobe and sidelobe structure associated with a finite-length linear aperture. For $\beta\in \{\pi/6, \pi/3\}$, bending alters the element-wise path-length profile and introduces a curvature-dependent phase component in the array response. In the near field ($r=100$ m), this component is appreciable and the $(\theta,\phi)$ pattern exhibits a distinct localized maximum in Fig.~\ref{fig:beamforming_all}(b)--Fig.~\ref{fig:beamforming_all}(c). Increasing $\beta$ sharpens this maximum and contracts the high-gain region. In the far field ($r=10^{4}$ m), the wavefront is effectively planar over the aperture, so the response becomes angle-dominated and the localized peak disappears in Fig.~\ref{fig:beamforming_all}(e)--Fig.~\ref{fig:beamforming_all}(f).

\begin{figure}[th]
    \centering      
    \begin{subfigure}{0.24\textwidth}  
        \centering
        \includegraphics[width=\textwidth]{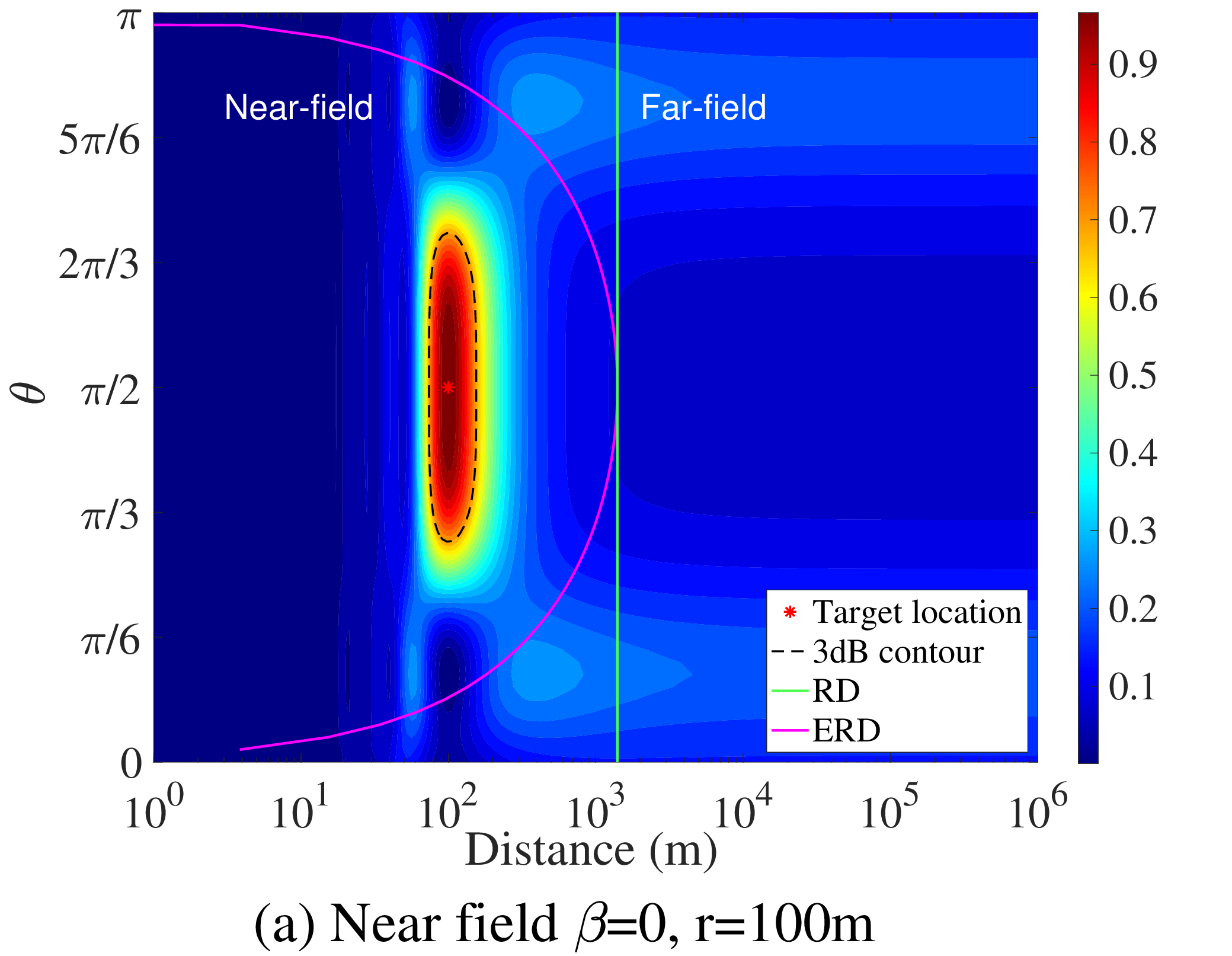}   
    \end{subfigure}    
    \hfill 
        \begin{subfigure}{0.24\textwidth}  
        \centering
        \includegraphics[ width=\textwidth]{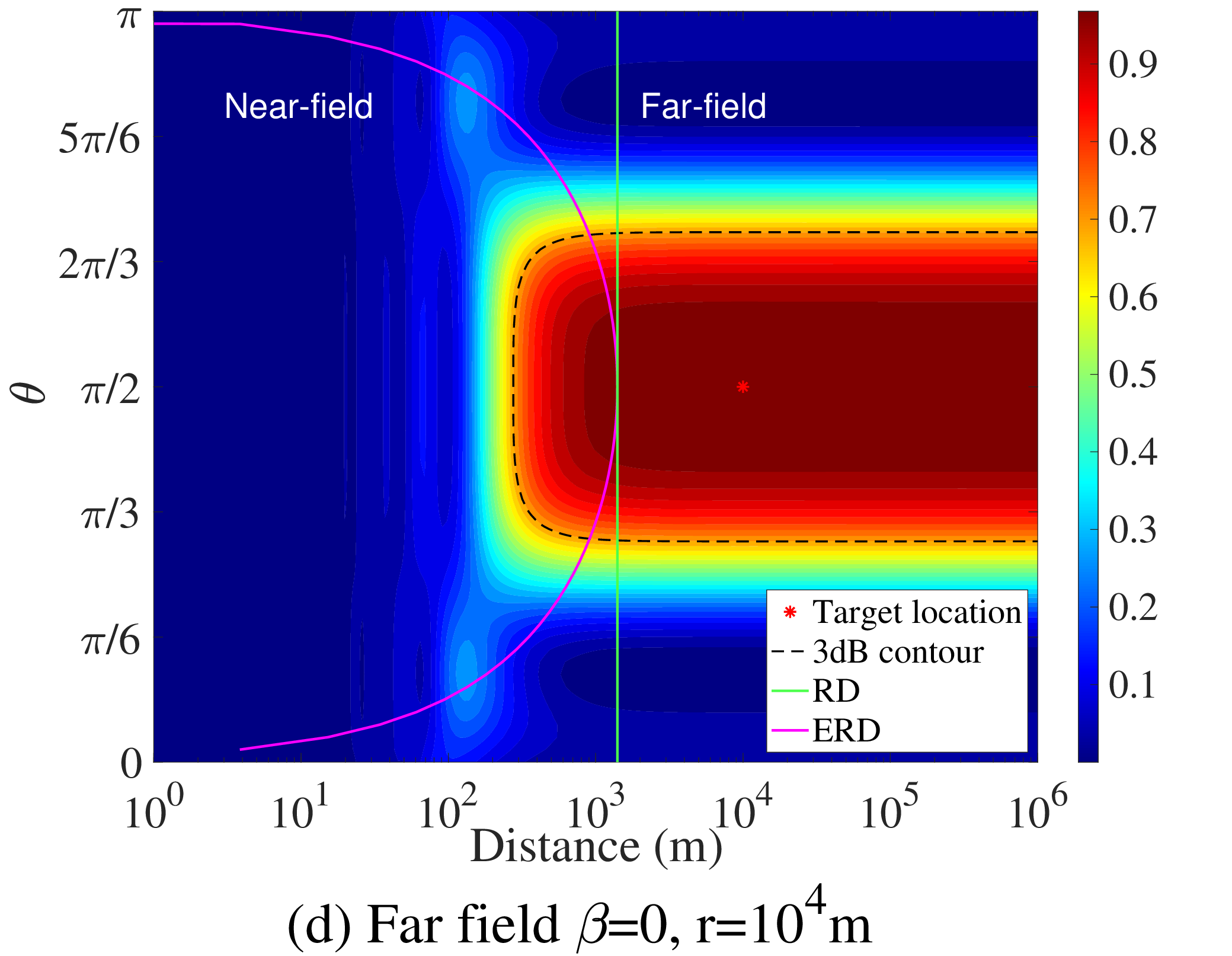}   
    \end{subfigure}    
    \hfill 
    \begin{subfigure}{0.24\textwidth} 
        \centering
        \includegraphics[width=\textwidth]{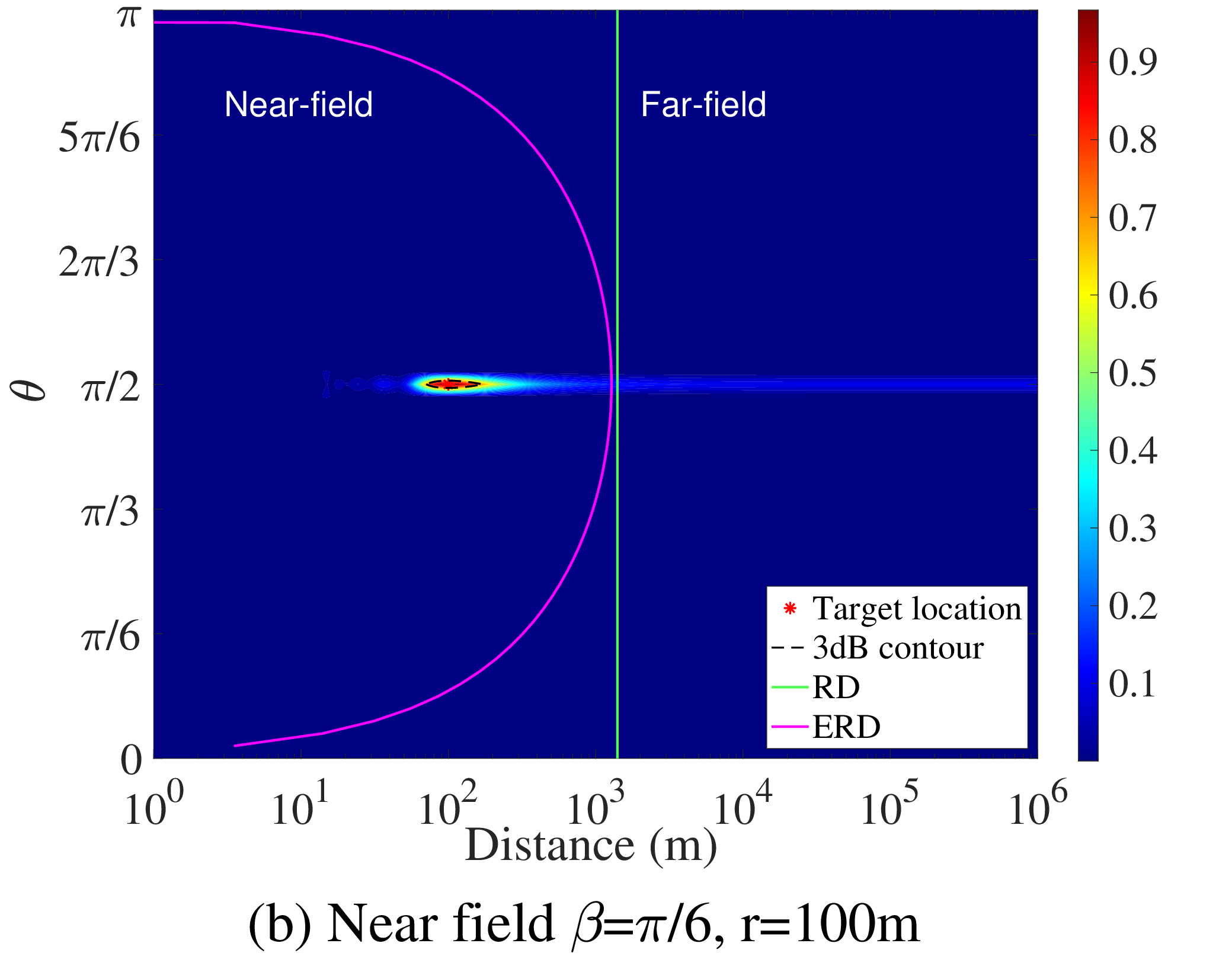}   
    \end{subfigure}
  \centering      
    \begin{subfigure}{0.24\textwidth}  
        \centering
        \includegraphics[ width=\textwidth]{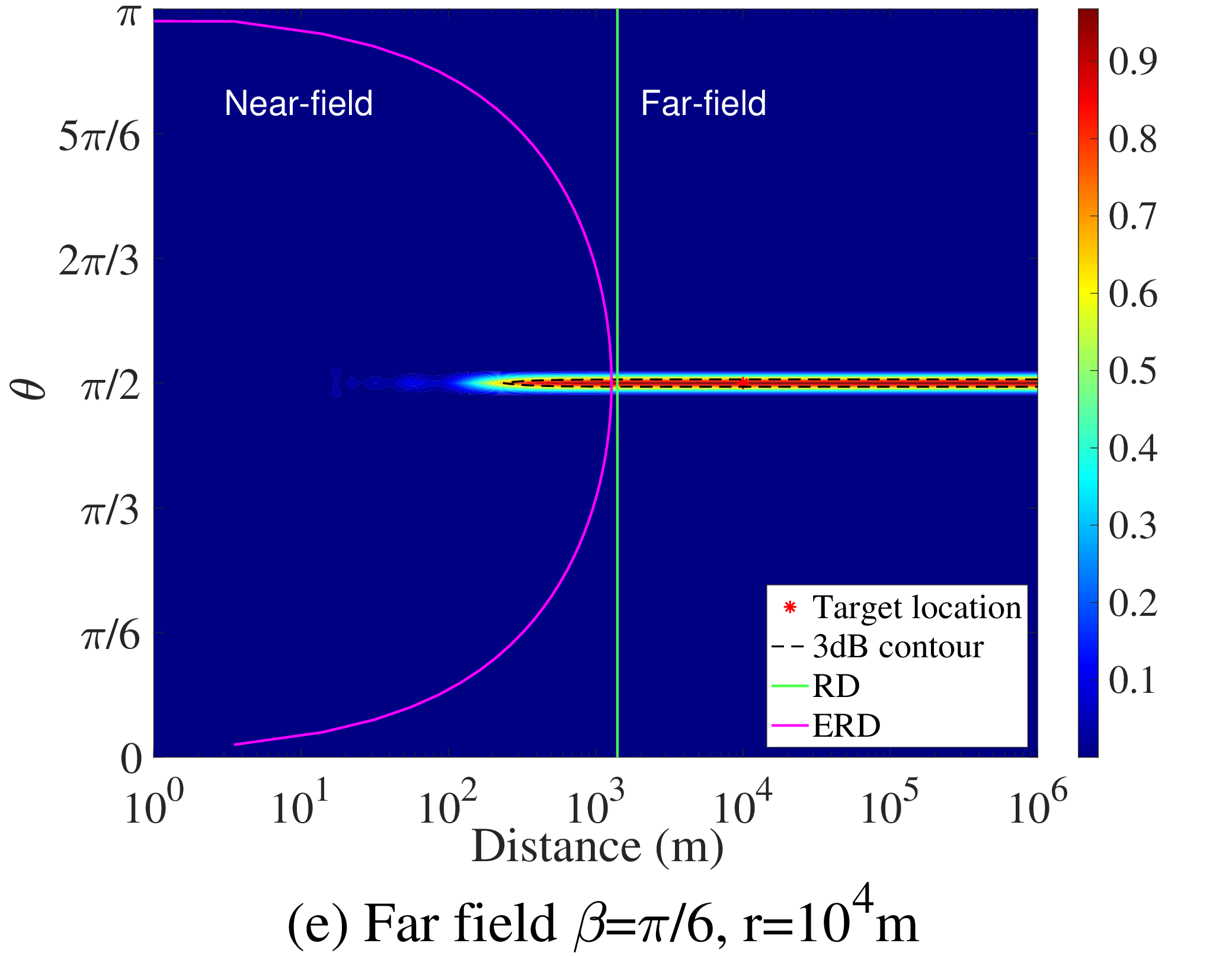}   
    \end{subfigure}    
   \hfill 
        \begin{subfigure}{0.24\textwidth}  
        \centering
        \includegraphics[width=\textwidth]{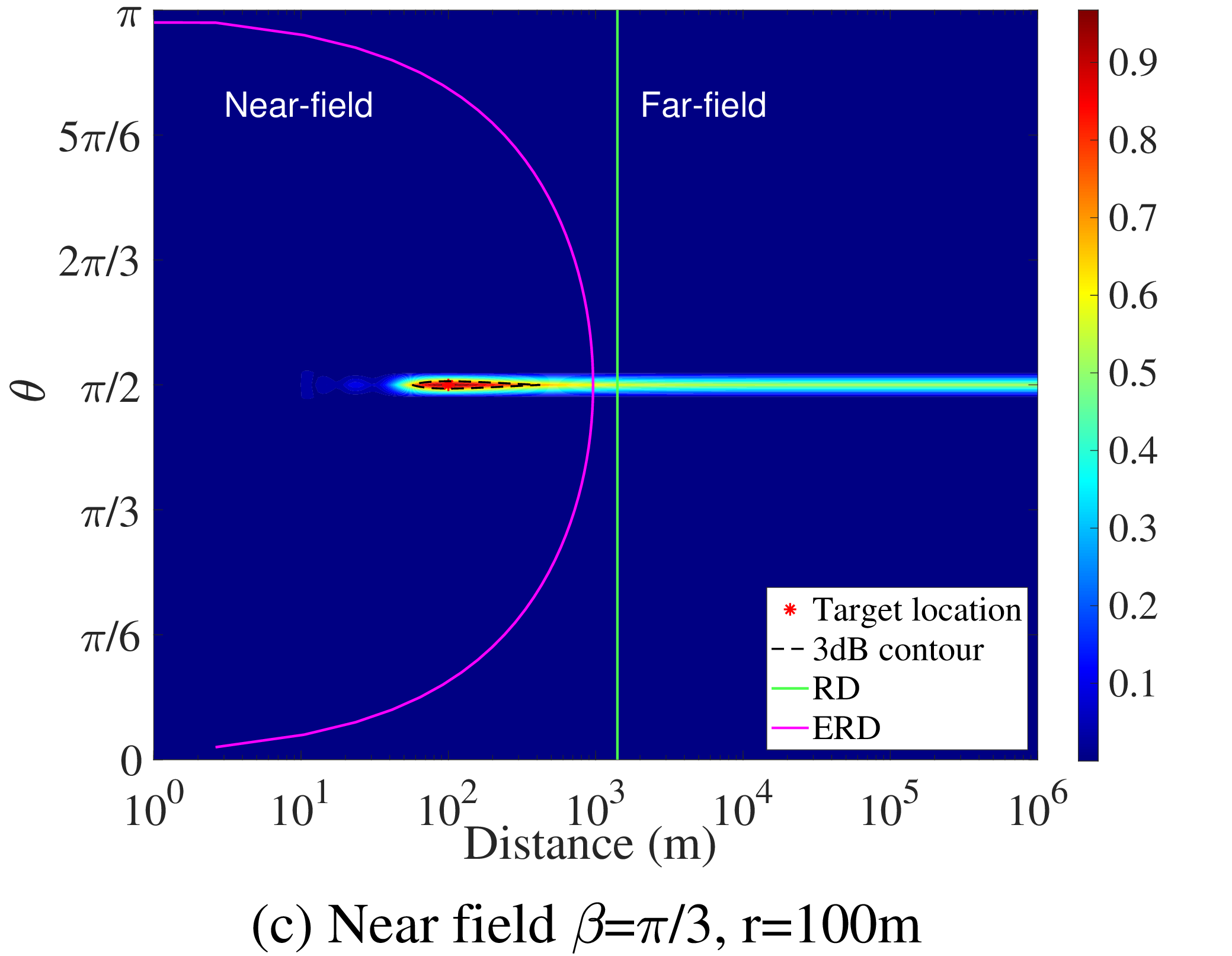}   
    \end{subfigure}    
    \hfill 
    \begin{subfigure}{0.24\textwidth} 
        \centering
        \includegraphics[ width=\textwidth]{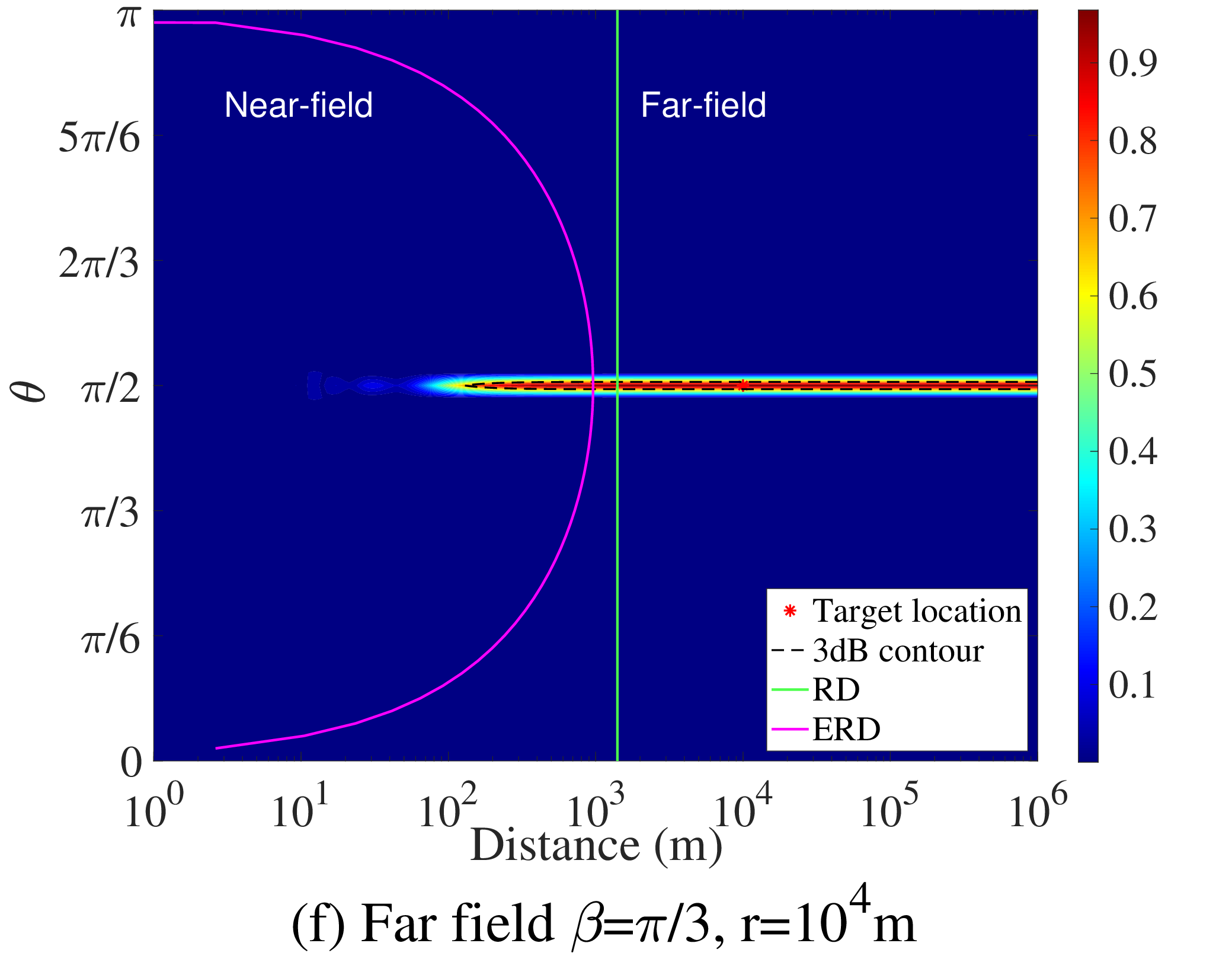}   
    \end{subfigure}  
    \caption{Beamforming gain on the range--elevation plane $(r,\theta)$ for the 1-D CuRA with fixed azimuth $\phi=0$. }    
    \label{fig:distance domain}   
        \vspace{-1.5 em}
\end{figure}

Fig.~\ref{fig:distance domain} shows the beamforming gain on the range--elevation plane $(r,\theta)$ for the 1-D CuRA with a fixed azimuth $\phi=0$. The red marker indicates the intended focal point, the black dashed curve is the $3$-dB contour, the green vertical line is the conventional RD, and the red curve denotes the proposed ERD. Fig.~\ref{fig:distance domain}(a)--Fig.~\ref{fig:distance domain}(c) correspond to near field focusing at $r =100$ m. For $\beta=0$ (Fig.~\ref{fig:distance domain}(a)), the high-gain region is localized around the target and the iso-gain contour forms a closed focusing cell in the $(r,\theta)$ plane. When the aperture is moderately bent to $\beta=\pi/6$ (Fig.~\ref{fig:distance domain}(b)), the $3$-dB contour becomes more compact in $\theta$, indicating a tighter angular selectivity. For a larger bending angle $\beta=\pi/3$ (Fig.~\ref{fig:distance domain}(c)), the contour becomes elongated along the $r$ axis, which suggests a reduced range selectivity. This non-monotonic evolution of the focusing morphology with respect to $\beta$ is consistent with the curvature-dependent range sensitivity predicted by the analytical characterization. Figs.~\ref{fig:distance domain}(d)--Figs.\ref{fig:distance domain}(f) show the far field counterparts with the beam steered toward $r = 10^{4}$ m. In this regime, the gain is nearly invariant with $r$ and the iso-gain contours collapse into strips parallel to the range axis, with the $3$-dB boundary appearing approximately as vertical lines. Increasing $\beta$ mainly sharpens the angular selectivity (narrower spread in $\theta$), while the range-focusing effect vanishes, which explains the similar contour patterns observed in Fig.~\ref{fig:distance domain}(e) and Fig.~\ref{fig:distance domain}(f). Unlike the RD that provides a constant transition distance, the ERD varies with $\theta$ and better matches the observed transition from a localized focusing cell to an $r$-invariant pattern.

\begin{figure}[th]
    \centering      
    \begin{subfigure}{0.24\textwidth}  
        \centering
        \includegraphics[width=\textwidth]{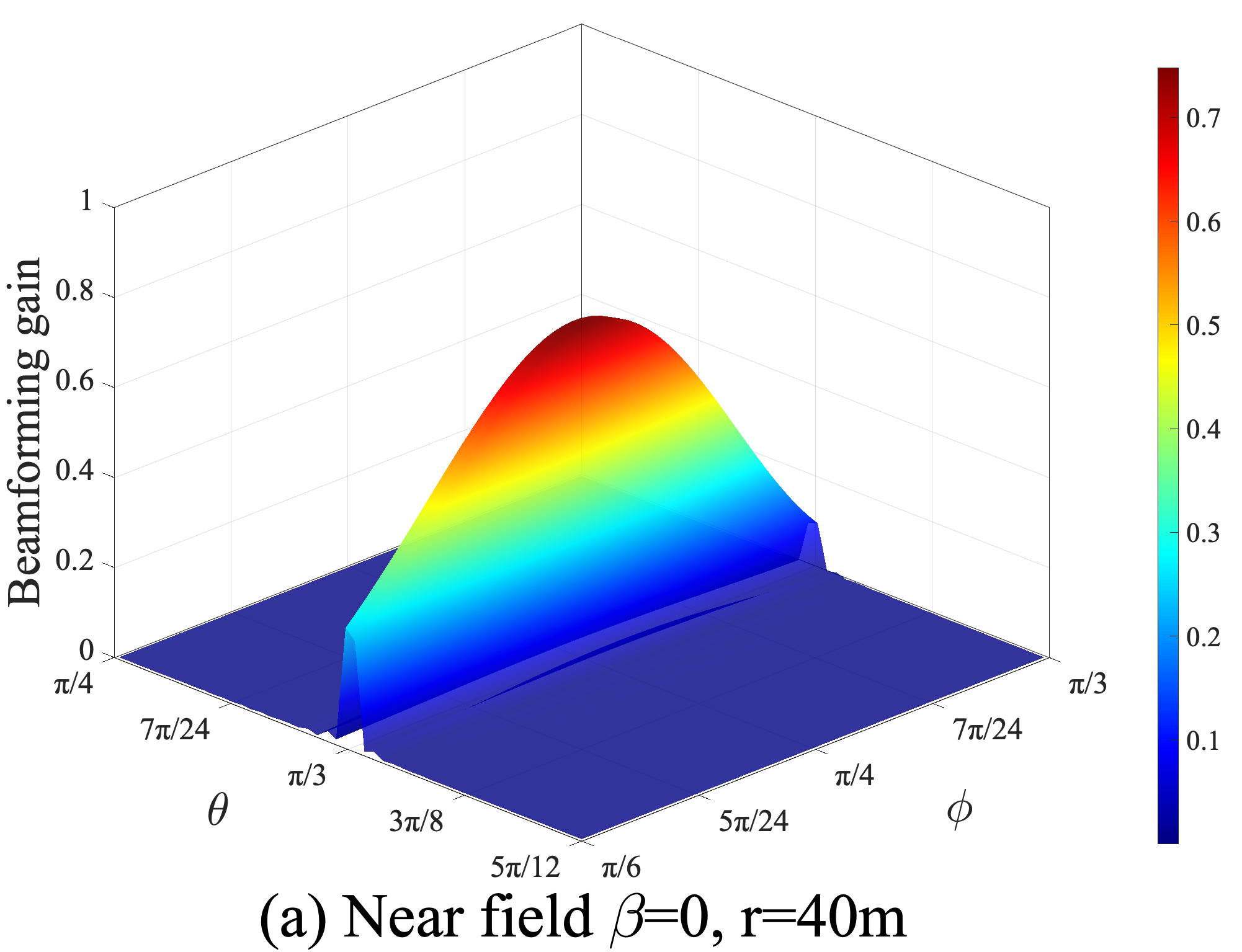 }   
    \end{subfigure}    
    \hfill 
        \begin{subfigure}{0.24\textwidth}  
        \centering
        \includegraphics[ width=\textwidth]{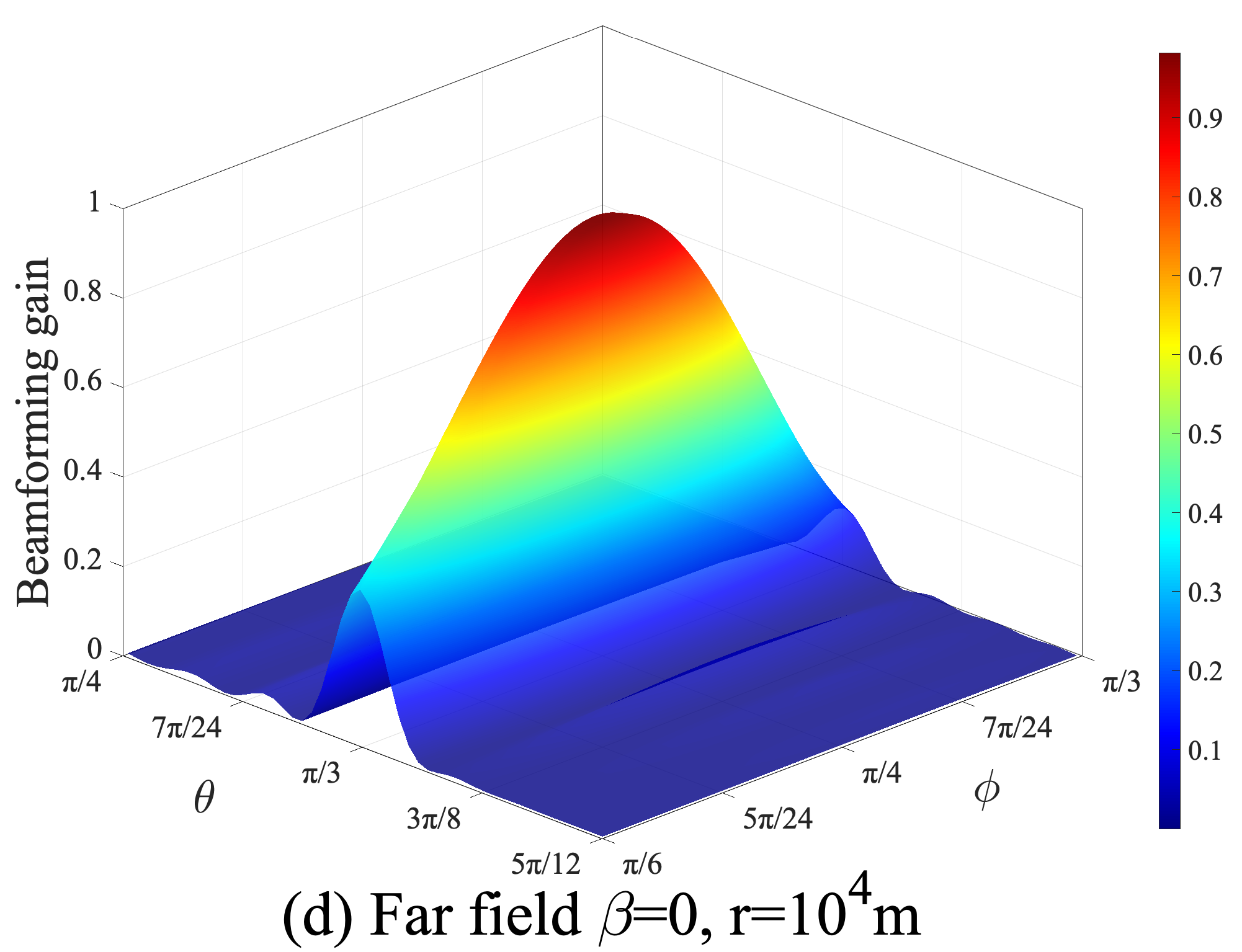 }   
    \end{subfigure}    
    \hfill 
    \begin{subfigure}{0.24\textwidth} 
        \centering
        \includegraphics[width=\textwidth]{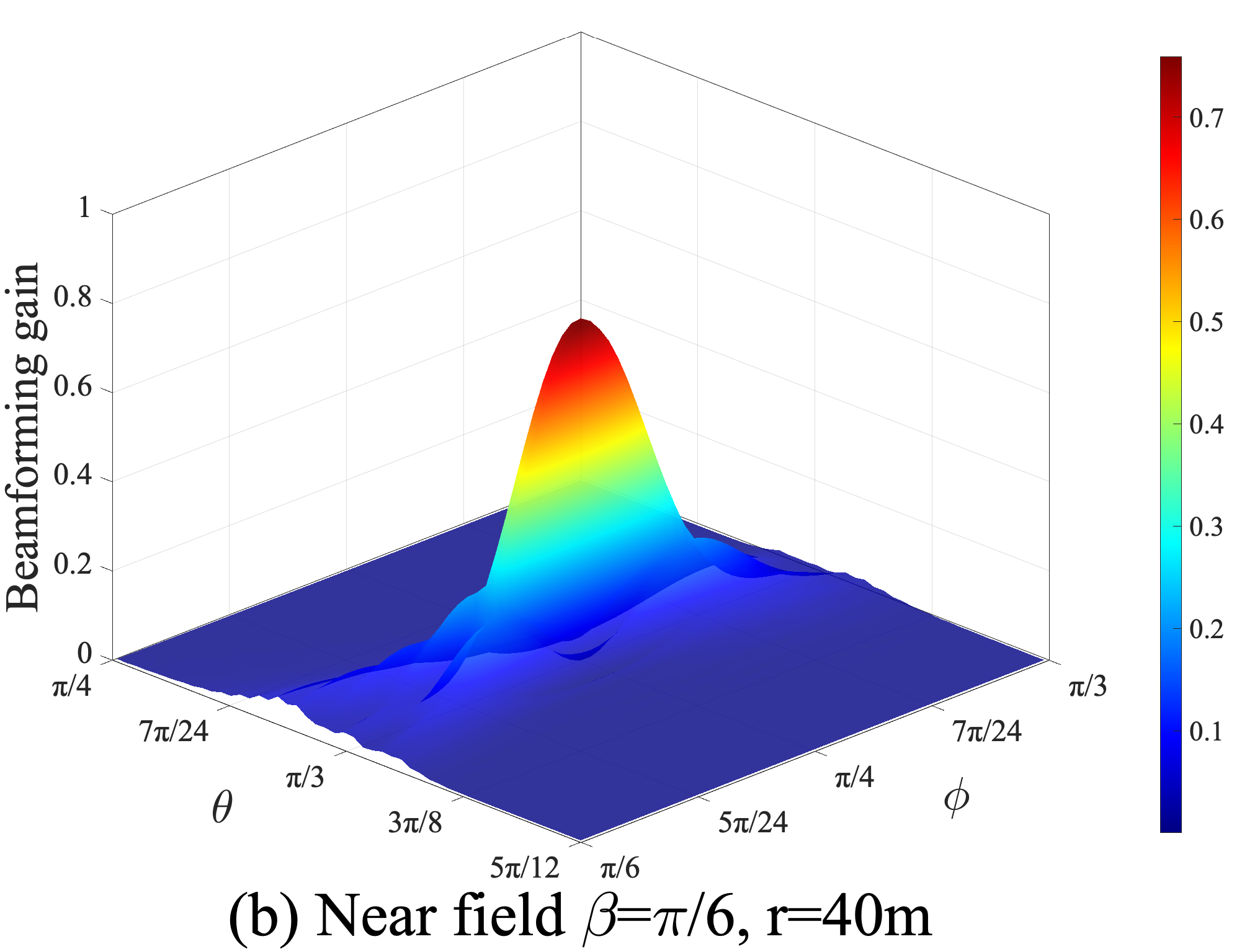 }   
    \end{subfigure}
  \centering      
    \begin{subfigure}{0.24\textwidth}  
        \centering
        \includegraphics[ width=\textwidth]{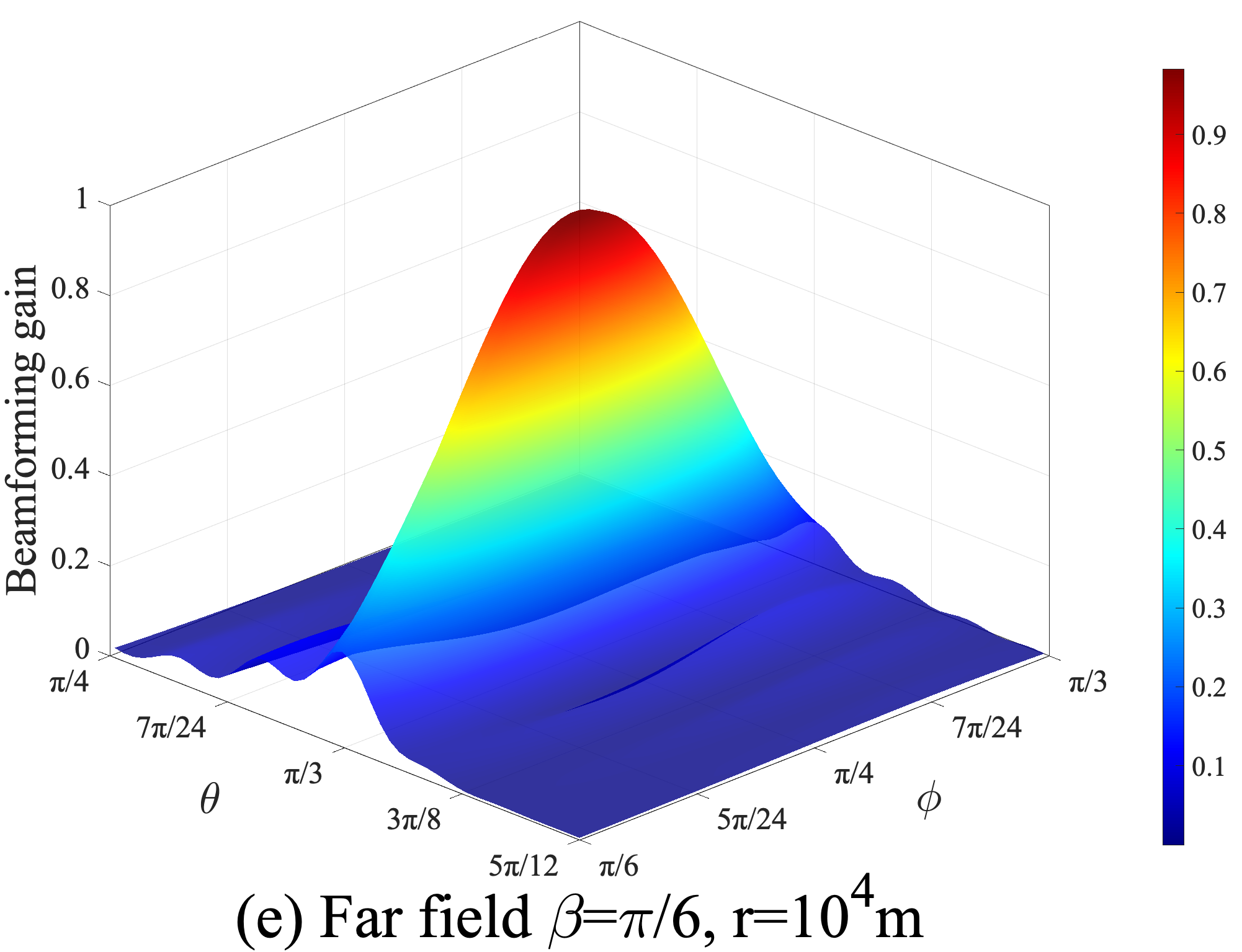}   
    \end{subfigure}    
    \hfill 
        \begin{subfigure}{0.24\textwidth}  
        \centering
        \includegraphics[width=\textwidth]{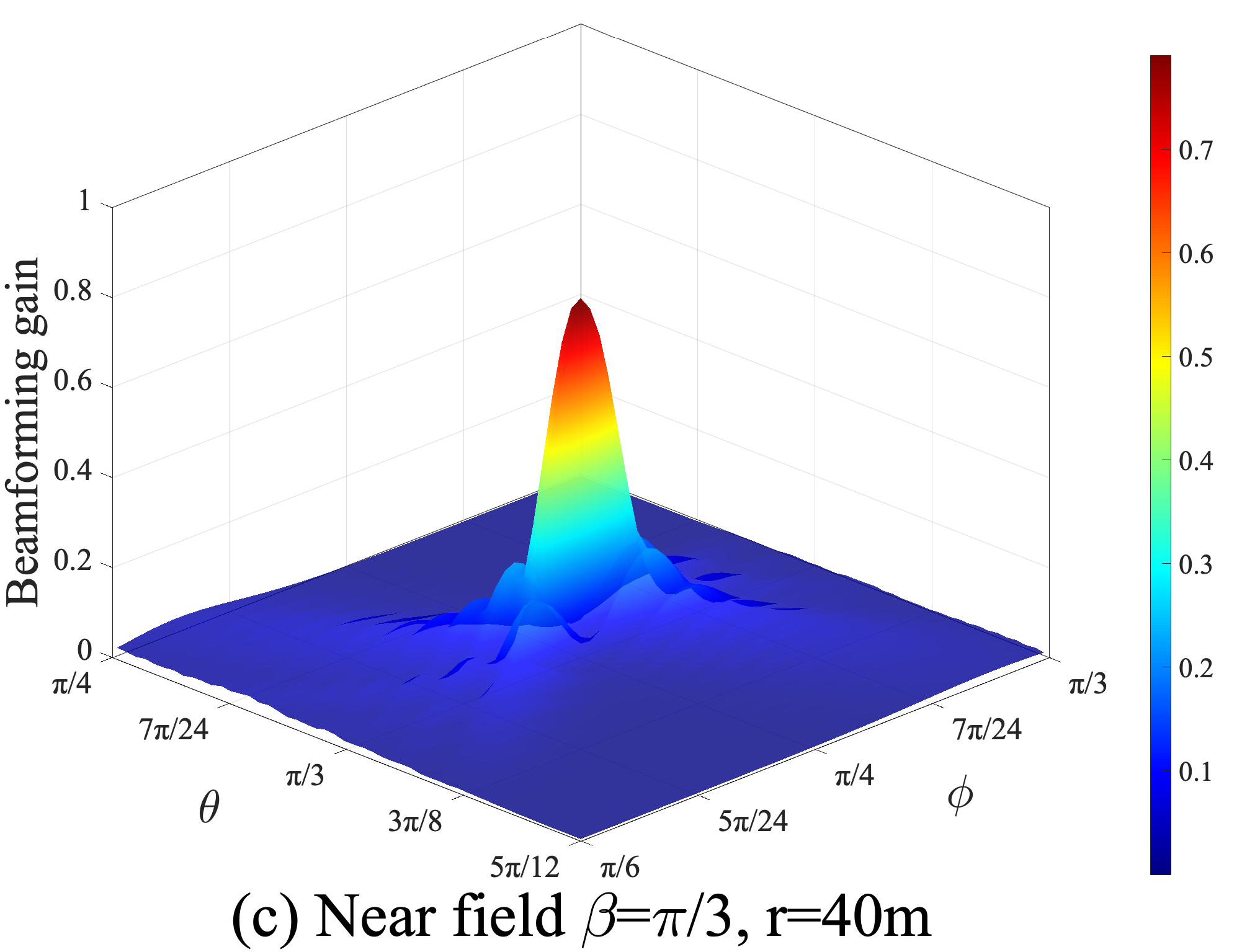}   
    \end{subfigure}    
    \hfill 
    \begin{subfigure}{0.24\textwidth} 
        \centering
        \includegraphics[ width=\textwidth]{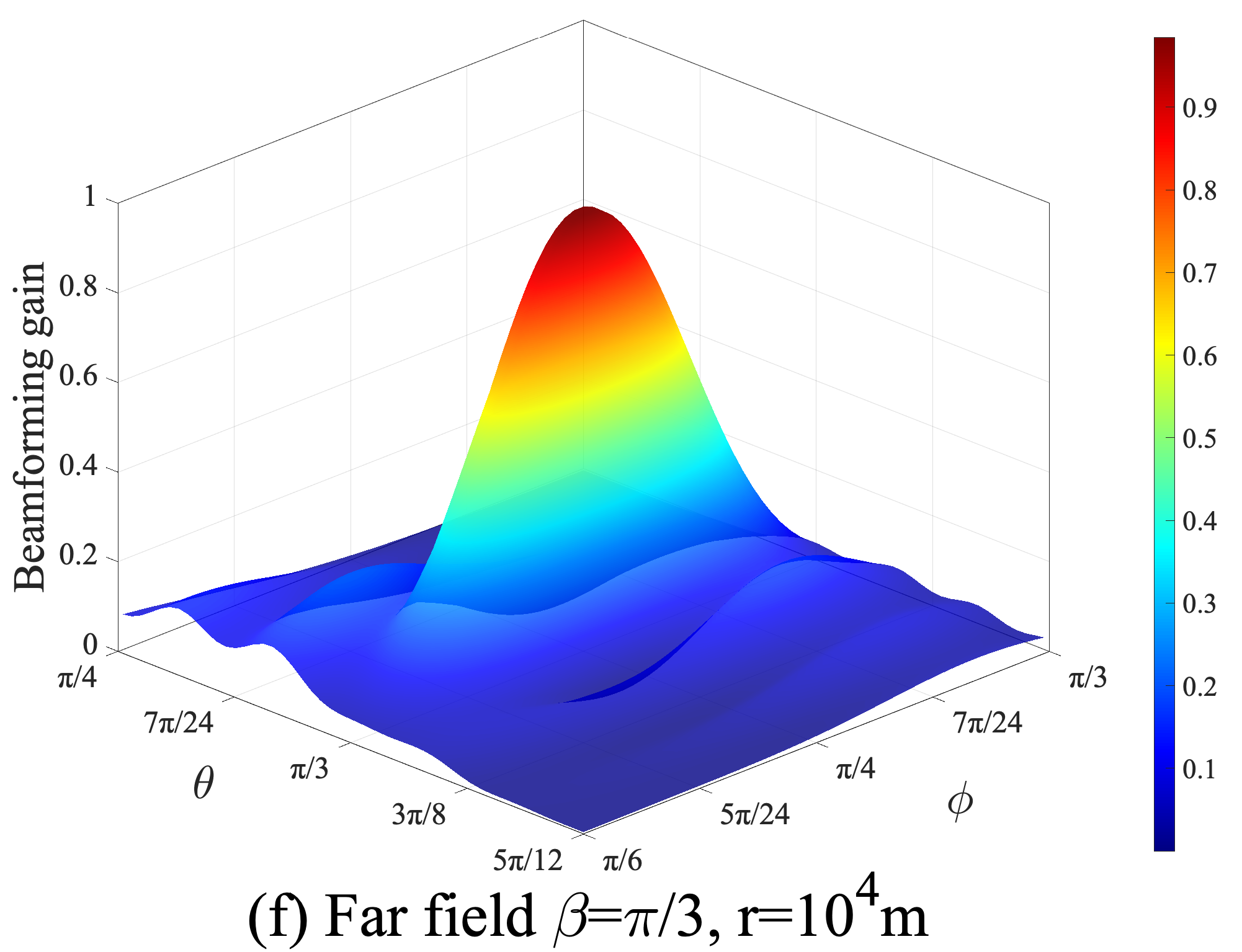}   
    \end{subfigure}  
    \caption{Beamforming gain over the angular domain $(\theta,\phi)$ for the 2-D CuRA with $M=8$, $N=64$, $\theta=\pi/2$.}
    \label{fig_2DCuRA}  
        \vspace{-1.5 em}
\end{figure}
 
Fig.~\ref{fig_2DCuRA} shows the beamforming gain of the 2-D CuRA ($M=8$, $N=64$) over the $(\theta,\phi)$ domain for $\beta\in\{0, \pi/6, \pi/3\}$. Figs.~\ref{fig_2DCuRA}(a)–(c) correspond to the near field case with $r=40$ m, while Figs.~\ref{fig_2DCuRA}(d)–(f) report the far field case with $r=10^{4}$ m. In the near field, a localized maximum is observed around the intended angular location for all three $\beta$. As $\beta$ increases from $0$ to $\pi/6$ and further to $\pi/3$, the high-gain region becomes more concentrated and the mainlobe footprint in the $(\theta,\phi)$ plane shrinks. Notably, the pattern exhibits a clear dependence on both $\theta$ and $\phi$, reflecting the two-dimensional aperture. This is in contrast to the 1-D setting, where the response typically shows weak variation along one angular dimension for the same visualization. In the far field, the patterns in Figs.~\ref{fig_2DCuRA}(d)–(f) are largely governed by angle-only steering. Compared with the near field plots, the sharpening effect with respect to $\beta$ is much less pronounced in this regime. The three far field patterns are visually similar, suggesting limited sensitivity to $\beta$ at large range. 

Fig.~\ref{fig_2DCuRAdistance} shows the beamforming gain of the 2-D CuRA on the $(r,\phi)$ domain, with $\theta = \pi/2$. Figs.~(a)–(c) show the near field focusing case with the target at $r=40$ m, and Figs.~(d)–(f) show the far field case with the target at $r=10^{4}$ m. In the near field, the maximum gain is concentrated around the target azimuth near $\phi\approx\pi/2$ and the focal range near $r\approx40$ m, and the $3$-dB contour encloses a localized high-gain region in the $(r,\phi)$ plane. From $\beta=0$ to $\beta=\pi/6$, the enclosed region becomes narrower in $\phi$ while remaining centered around the focal distance. For $\beta=\pi/3$, the contour remains narrow in $\phi$ and exhibits a larger spread along $r$ compared with $\beta=\pi/6$. The RD is constant across $\phi$, whereas the ERD varies with $\phi$ and shifts the near-/far field boundary accordingly. In the far field, the gain becomes nearly invariant with respect to $r$ around the steered azimuth. The three far field patterns appear close to each other, with only minor changes as $\beta$ varies.

\begin{figure}[!th]
    \centering      
    \begin{subfigure}{0.24\textwidth}  
        \centering
        \includegraphics[ width=\textwidth]{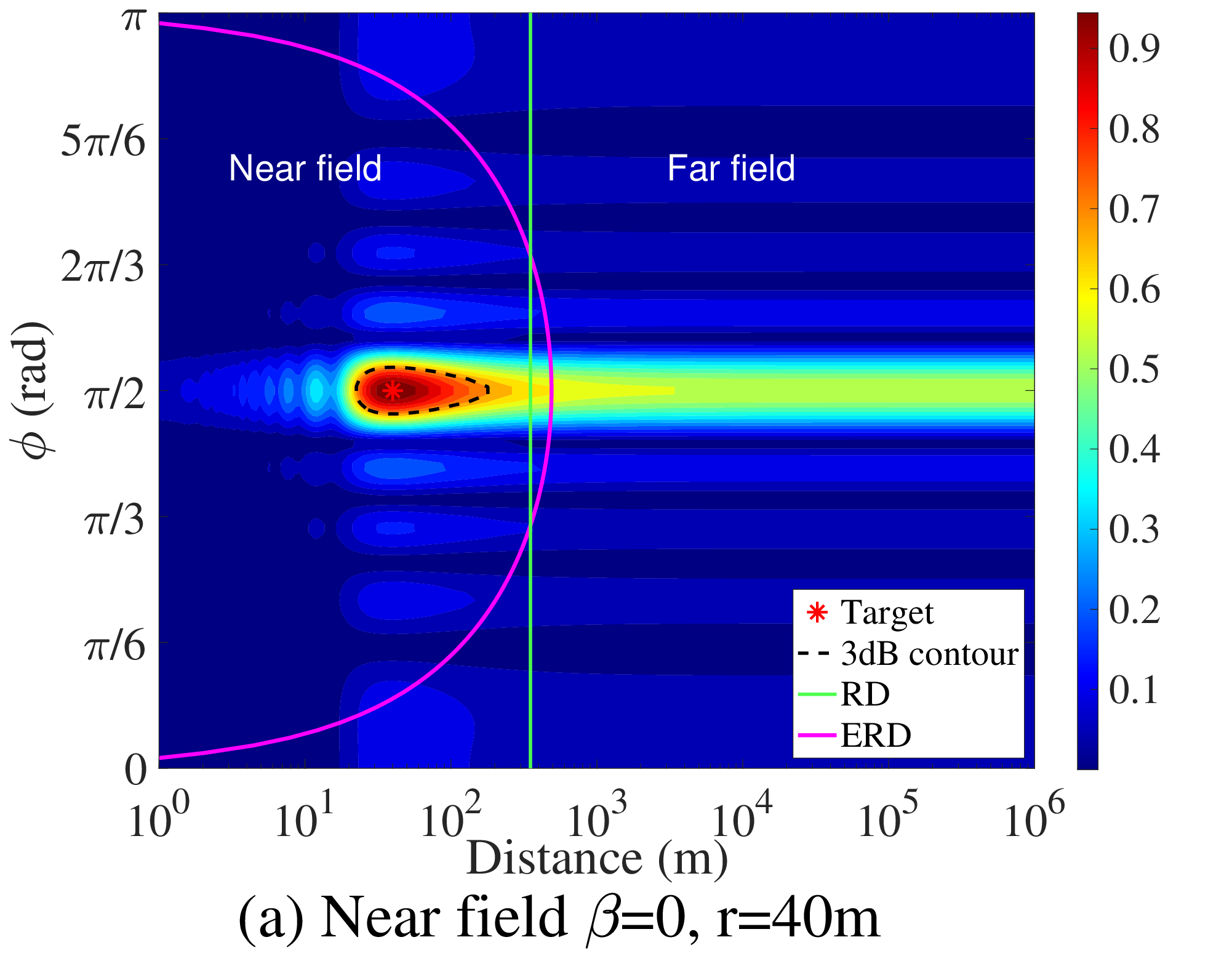}   
    \end{subfigure}    
    \hfill 
        \begin{subfigure}{0.24\textwidth}  
        \centering
        \includegraphics[  width=\textwidth]{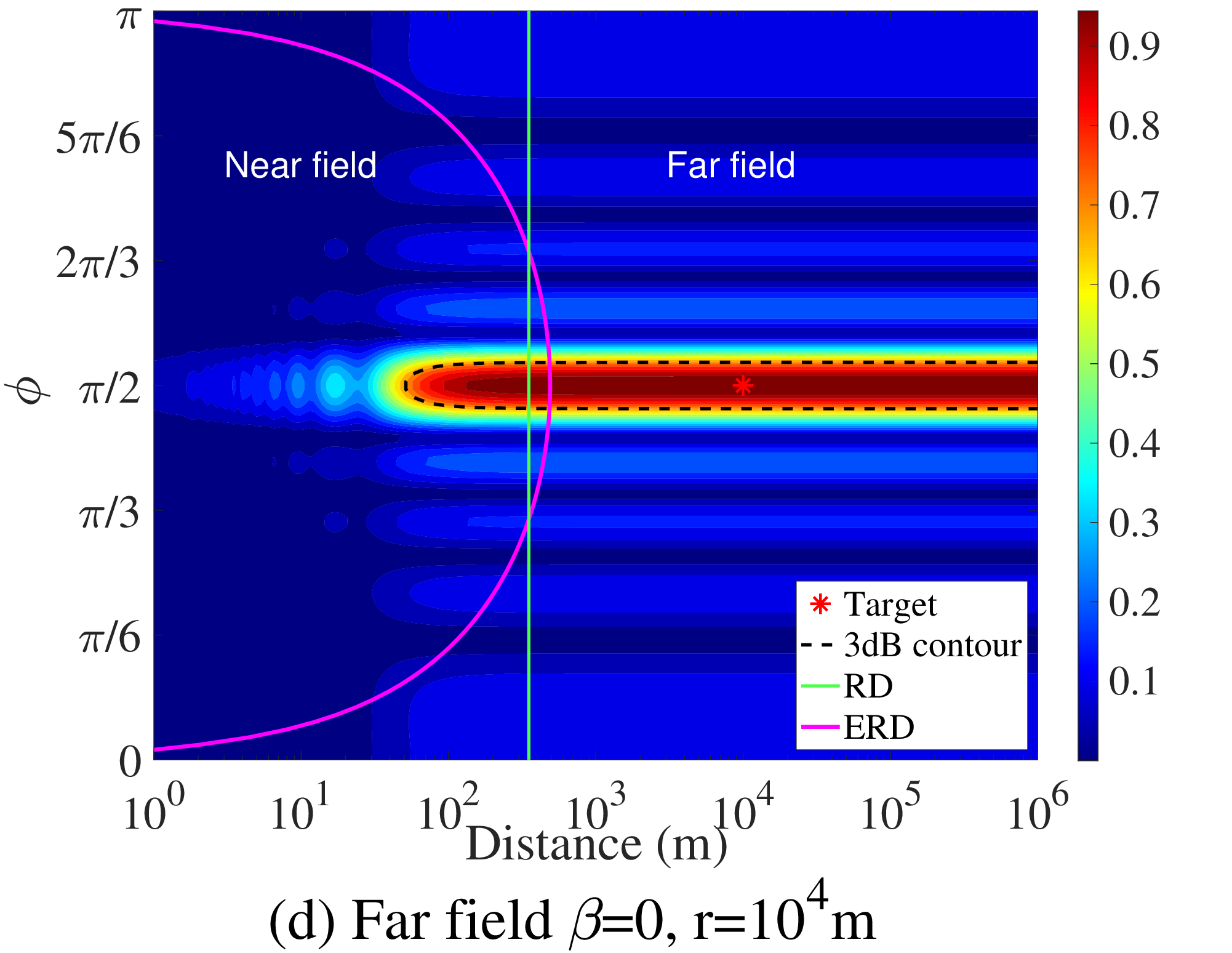} 
    \end{subfigure}    
    \hfill 
    \begin{subfigure}{0.24\textwidth} 
        \centering
        \includegraphics[  width=\textwidth]{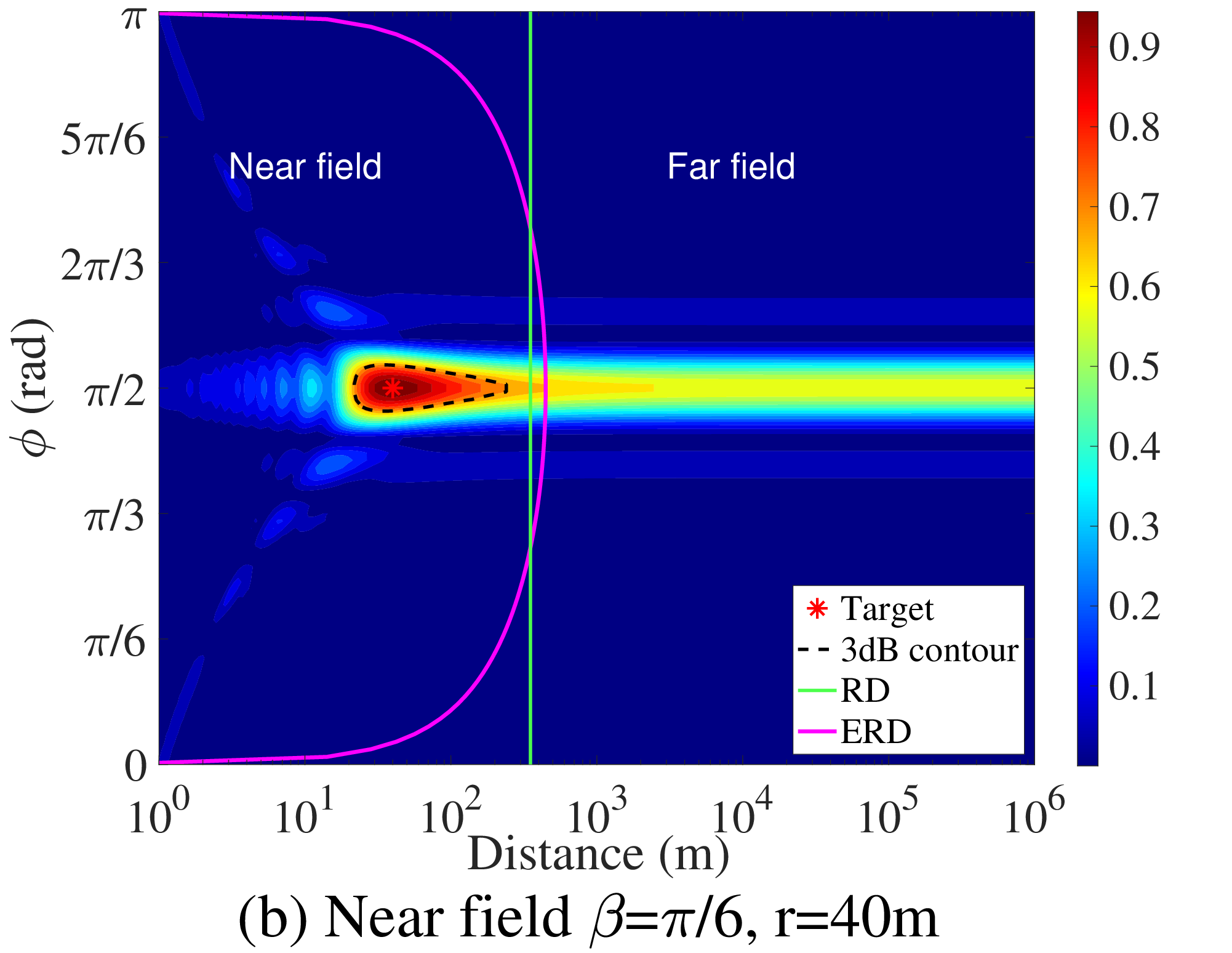} 
    \end{subfigure}
  \centering      
    \begin{subfigure}{0.24\textwidth}  
        \centering 
             \includegraphics[width=\textwidth]{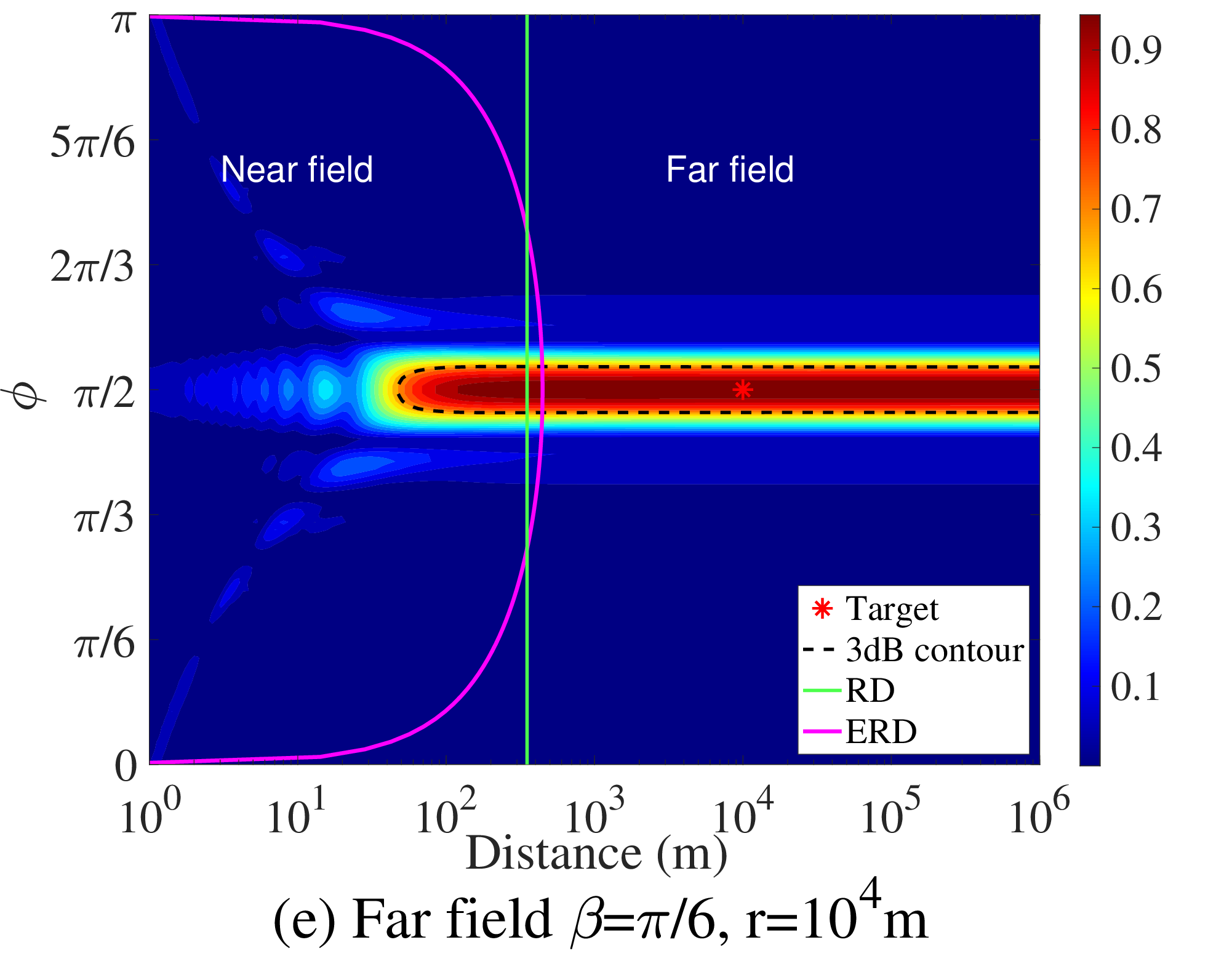} 
    \end{subfigure}    
    \hfill 
        \begin{subfigure}{0.24\textwidth}  
        \centering
             \includegraphics[width=\textwidth]{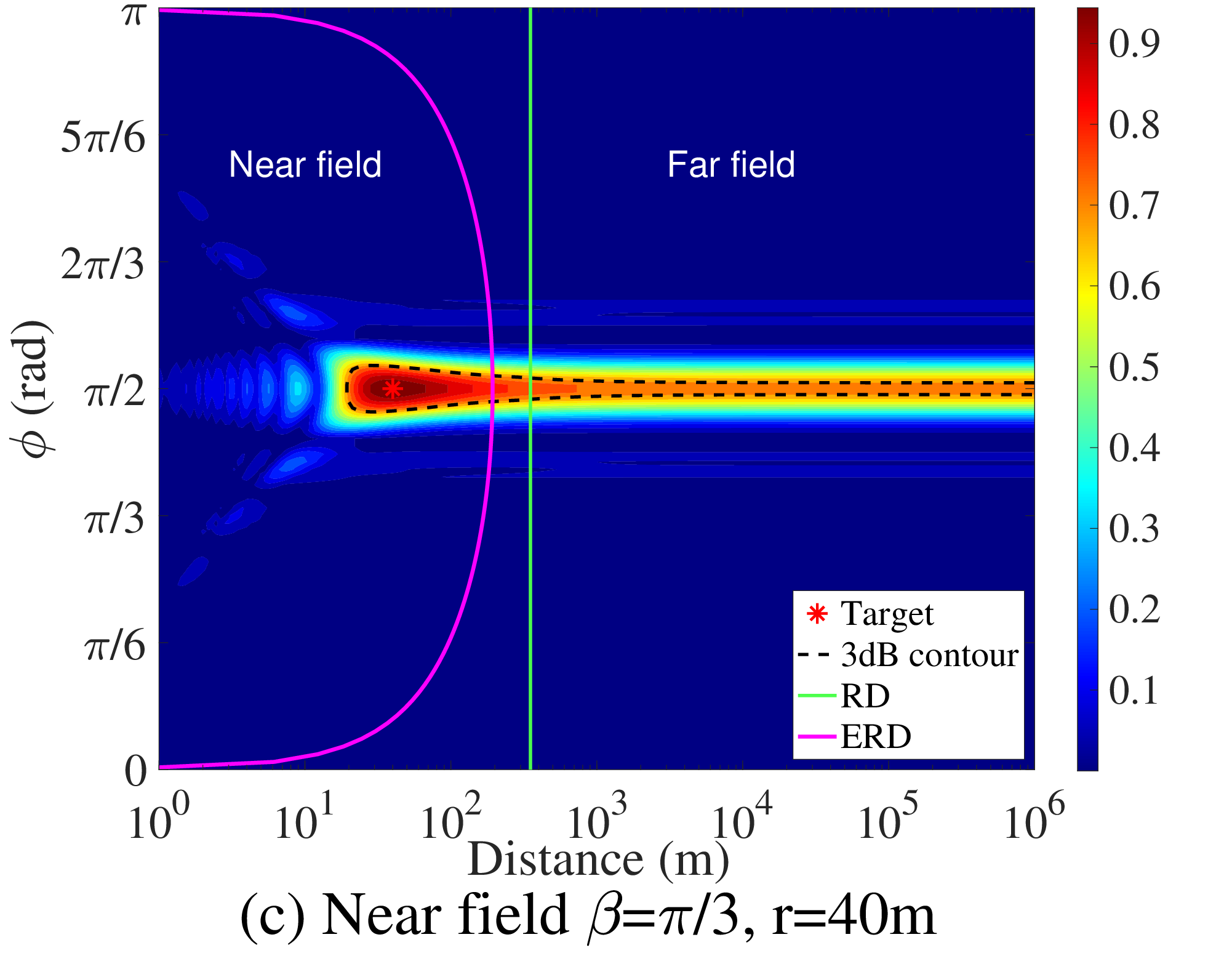} 
    \end{subfigure}    
    \hfill 
    \begin{subfigure}{0.24\textwidth} 
        \centering  
            \includegraphics[width=\textwidth]{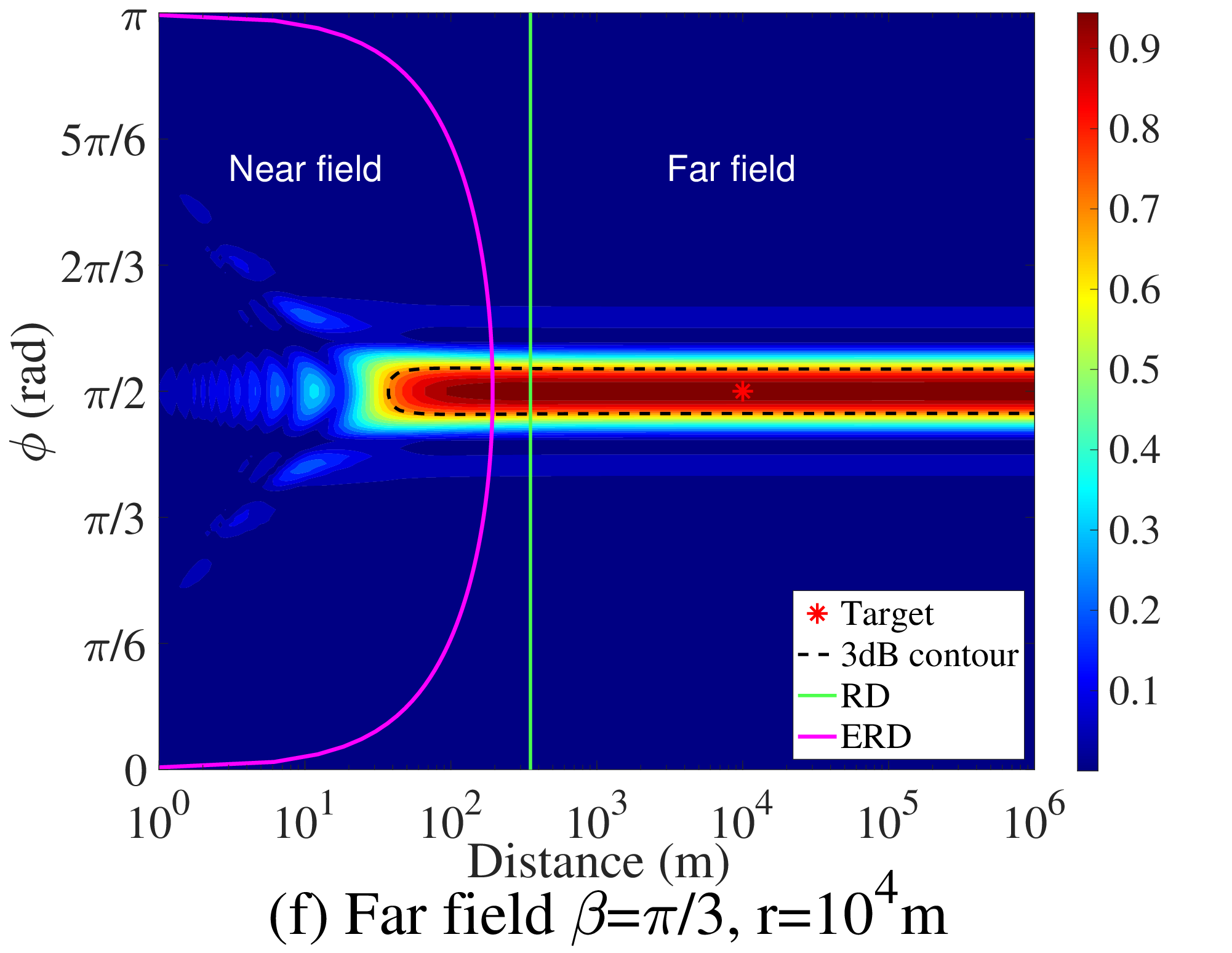}
    \end{subfigure}  
    \caption{ Beamforming gain on the range--azimuth plane $(r,\phi)$ for the 2-D CuRA.}
    \label{fig_2DCuRAdistance}  
        \vspace{-1.5 em}
\end{figure}

\begin{figure}[!t] 
    \centering
    \includegraphics[width=0.48\textwidth]{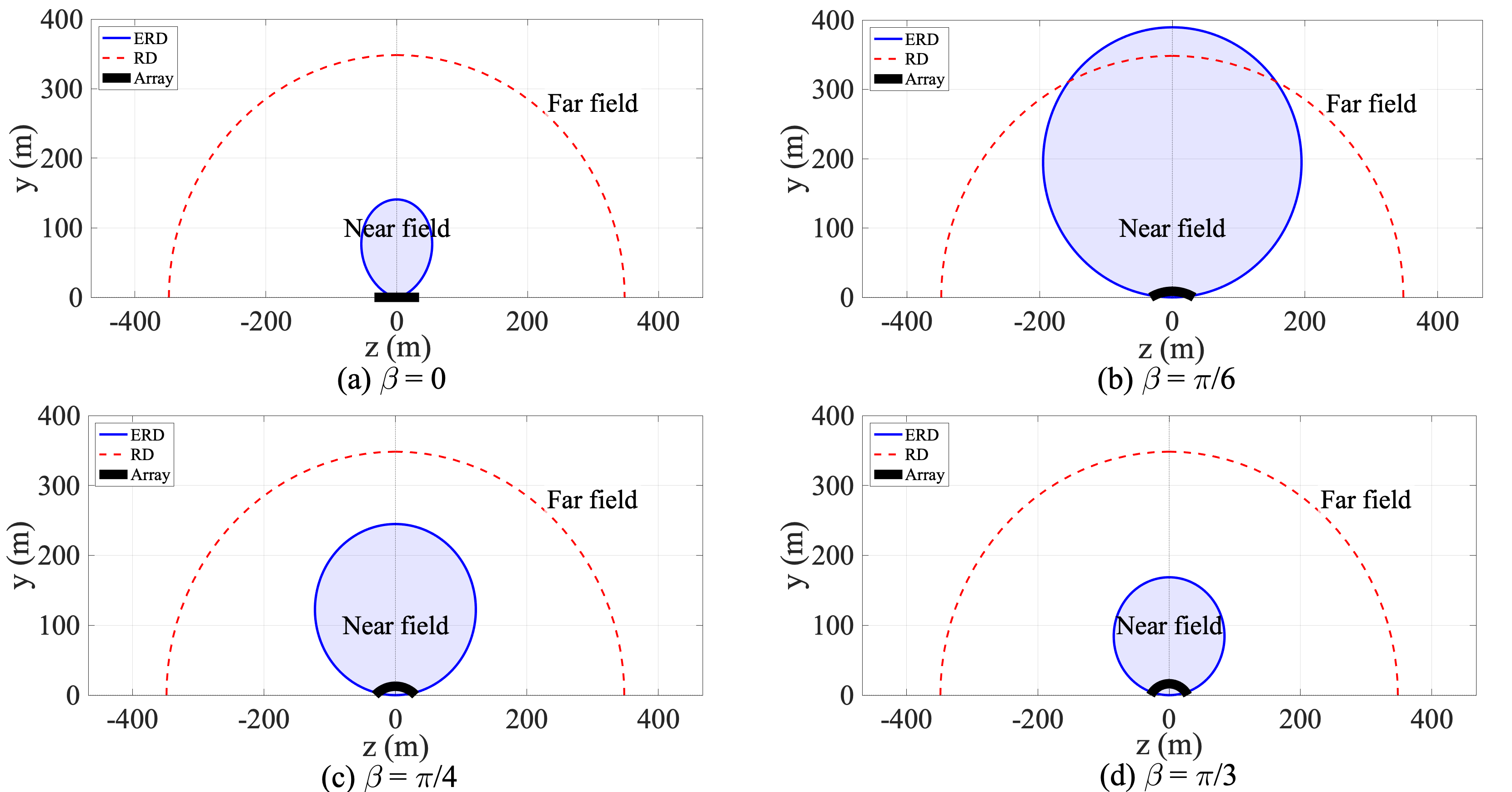}
    \caption{Near-field region characterized by the ERD on the $yz$-plane for different bending angles.}
    \label{fig:ERD}
        \vspace{-1.5 em}
\end{figure}

Fig.~\ref{fig:ERD} shows the ERD contour projected onto the $yz$-plane for different bending angles $\beta$, together with the conventional RD. The blue closed curve indicates the ERD-defined near field region, while the red dashed arc corresponds to the RD boundary, which is direction-invariant.
For $\beta=0$ in Fig.~\ref{fig:ERD}(a), the CuRA reduces to a ULA. The ERD region is concentrated around the positive $y$-axis and shrinks rapidly toward the lateral directions along $z$. The ERD contour is also much smaller than the RD arc in this panel. When the array is bent to $\beta=\pi/6$ in Fig.~\ref{fig:ERD}(b), the ERD region expands in both the $y$ and $z$ directions compared with Fig.~\ref{fig:ERD}(a). For $\beta=\pi/4$ in Fig.~\ref{fig:ERD}(c), the ERD region remains extended along $y$, while its lateral spread along $z$ is narrower than that in Fig.~\ref{fig:ERD}(b). For $\beta=\pi/3$ in Fig.~\ref{fig:ERD}(d), the ERD region becomes more concentrated around the $y$-axis direction, with a further reduced extent along $z$ relative to Fig.~\ref{fig:ERD}(c).
It can be observed that the fixed RD arc cannot reflect the direction-dependent boundary variations shown by the ERD contours across $\beta$. This visualization is consistent with using ERD to partition the operating region, where range-aware (near field) beams are applied inside the ERD contour and angle-only (far field) beams are applied outside.

\begin{figure}[ht] 
    \centering
    \includegraphics[width=0.5\textwidth]{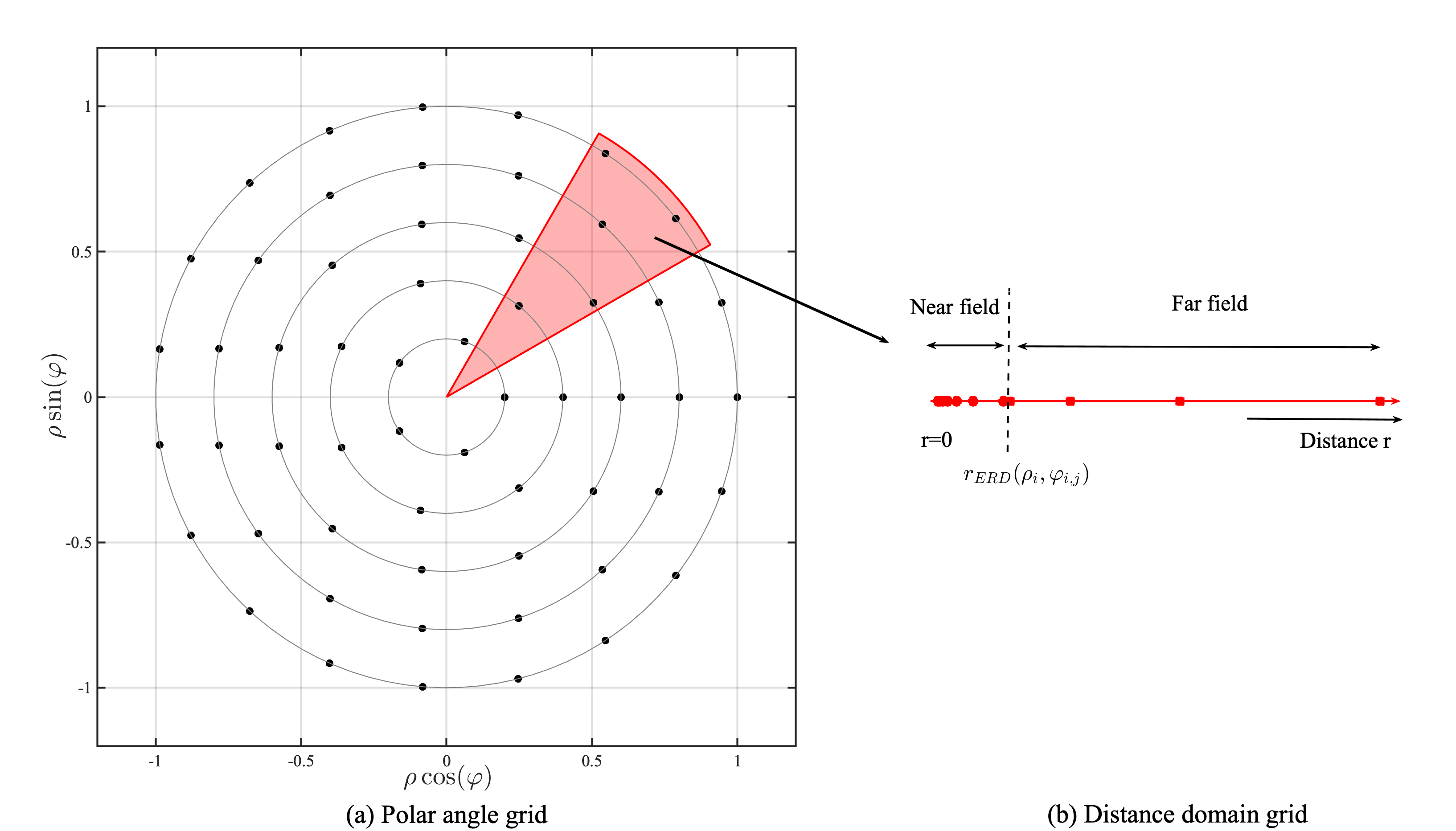}
    \caption{Illustration of the proposed codebook sampling strategy.}
    \label{fig:sample}

\end{figure}

Fig.~\ref{fig:sample} shows the sampling set generated by the proposed codebook construction. In Fig.~\ref{fig:sample}(a), the angular samples are placed on a polar-domain grid. The number of samples per $\rho$ ring increases with $\rho$, i.e., the grid is coarser at small $\rho$ and denser at large $\rho$, following the $\rho$-dependent angular resolution requirement. Fig.~\ref{fig:sample}(b) shows the range sampling for a fixed angular direction $(\rho_i,\varphi_{i,j})$. The ERD value $r_{\mathrm{ERD}}(\rho_i,\varphi_{i,j})$ defines a direction-dependent transition point on the range axis. Range samples are placed densely over $[0,,r_{\mathrm{ERD}}(\rho_i,\varphi_{i,j})]$ and sparsely for $r>r_{\mathrm{ERD}}(\rho_i,\varphi_{i,j})$. The overall codebook is obtained by combining the polar-domain angular grid with the ERD-gated range grid for each sampled direction. 
Notably, the proposed codebook adopts ERD-based sampling strategy in the distance domain. When $r$ is much larger than the ERD, the sampling interval increases, and the codebook automatically degenerates into a conventional far field codebook. Conversely, when $r$ lies within the ERD, the codebook performs adaptive dense sampling to ensure precise near field beam focusing. As a result, the codebook operates seamlessly without requiring switching between near field and far field regimes. It guarantees accurate focusing in the near field region while remaining compatible with conventional far field beamforming, thereby establishing a unified design framework that encompasses both near- and far- field scenarios.

\begin{figure}[!t]
    \centering 
     \includegraphics[width=0.5\textwidth]{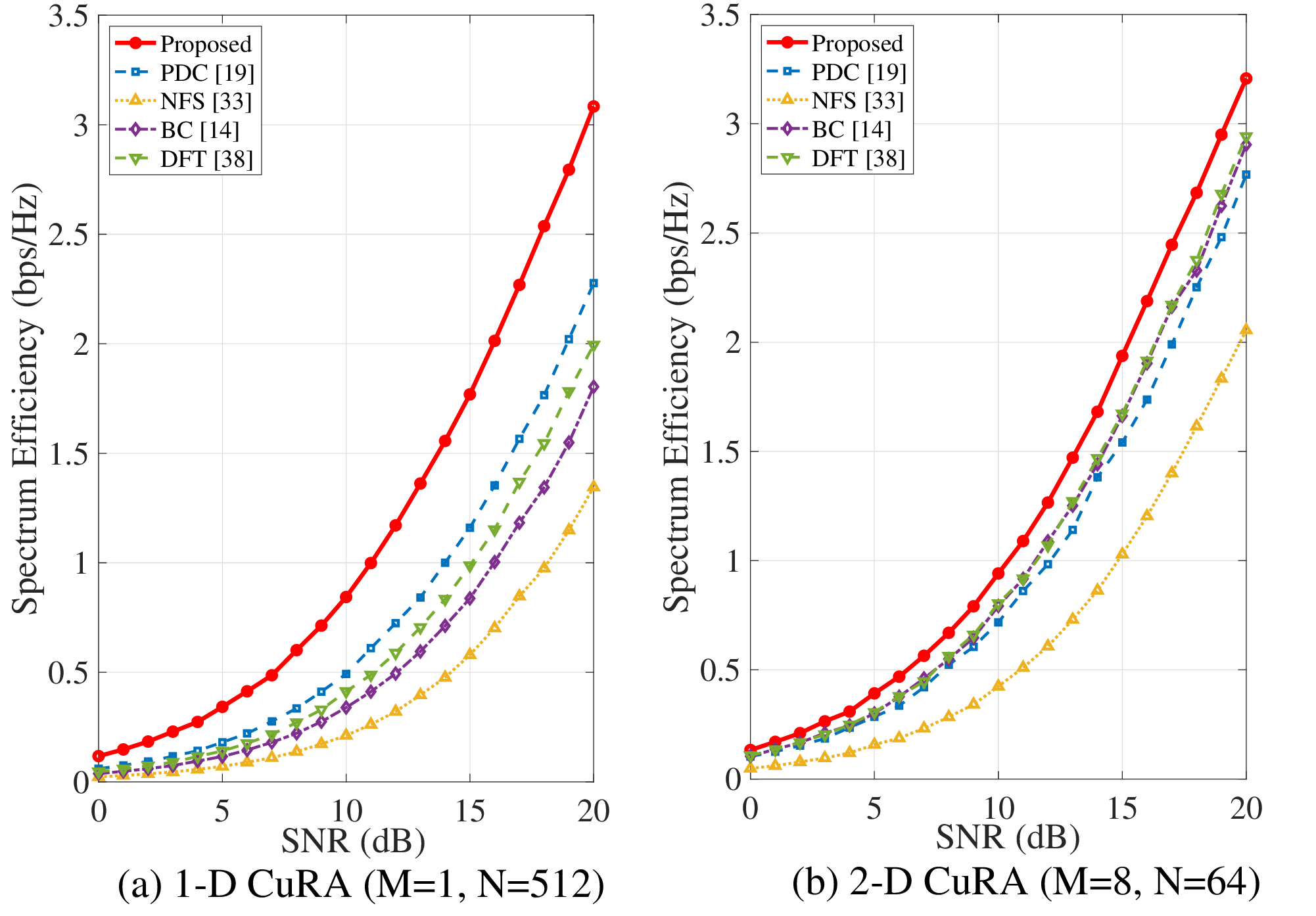} 
    \caption{Spectral efficiency versus SNR for different codebook schemes.}
    \label{fig_SE} 
        \vspace{-1.5 em}
\end{figure} 

Finally, we evaluate the spectral efficiency achieved by different codebook designs. For a single-stream link, the spectral efficiency is defined as
$\mathcal{E}=\log_2\!\left(1+\frac{P|\mathbf{h}^{\mathrm H}\mathbf{v}|^2}{\sigma^2}\right)$,
where $P$ is the transmit power, $\mathbf{v}$ denotes the selected codeword satisfying $\|\mathbf{v}\|_2=1$, $\mathbf{h}$ is the normalized channel vector with $\|\mathbf{h}\|_2=1$, and $\sigma^2$ is the noise power. We compare the proposed codebook with several representative benchmark schemes for both the 1-D CuRA and 2-D CuRA configurations, including the polar-domain codebook (PDC) in \cite{ula}, the near-field spherical codebook (NFS) in \cite{3Ds}, the beamforming codebook (BC) in \cite{cdbook}, and the classical DFT codebook in \cite{DFT}. As shown in Fig.~\ref{fig_SE}, under the same codebook size, the proposed design yields the highest spectral efficiency across the tested SNR range for both array configurations. The curves are close at low SNR, while the separation becomes more apparent as the SNR increases.

\section{Conclusion}\label{sec:conclusion}
We considered codebook-based cross near- and far-field beam training for CuRAs. By analyzing the beamforming-gain correlation under spherical-wave propagation, we derived a polar-domain angular sampling rule with $\rho$-dependent resolution. To capture the direction-dependent near-/far-field transition induced by the curvature, we introduced ERD and constructed an ERD-gated range grid, where the range sampling is dense inside the ERD region and sparse outside. The resulting hierarchical codebook provides explicit spatial coverage control and applies to both 1-D and 2-D apertures. Numerical results demonstrated consistent gains in spectral efficiency over several representative baselines under the same beam-search budget. Extensions to multipath channels, calibration and coupling effects under reconfiguration, and multiuser training protocols are of interest for future work.

\begin{appendices}
\section{}\label{app:proof_BG1}
Under the far field assumption, the array response vector for direction $(\rho, \varphi)$ can be expressed as 
\begin{align*}
    \mathbf{a}(\rho,\varphi) = \frac{1}{\sqrt{N}}\begin{bmatrix}
e^{j\Phi_1(\rho,\varphi)} &\cdots & e^{j\Phi_N(\rho,\varphi)}
\end{bmatrix}^T,\tag{A1}
\end{align*}
where $\Phi_n(\rho,\varphi)$ denote the phase of the $n$-th element as
\begin{align*}
    \Phi_n(\rho,\varphi) = \frac{2\pi R}{\lambda}&\left[\sin\rho\sin\varphi\left(\cos(\beta-\psi_n)-\cos\beta\right) \right. \\
& \left. + \cos\rho\cos\varphi\sin(\beta-\psi_n)\right].\tag{A2}
\end{align*}
 
Consider two directions $(\rho_1,\varphi_1)$ and $(\rho_2,\varphi_2)$, the beamforming gain can be expressed as 
\begin{align*}
   & G_{\text{angle}}(\mathbf{s}_1, \mathbf{s}_2)  \approx |\mathbf{a}^H(\rho_1,\varphi_1)\mathbf{a}(\rho_2,\varphi_2)| = \left|\frac{1}{N}\sum_{n=1}^{N} e^{j\Delta\Phi_n(\rho,\varphi)}\right|\\
    &\quad\approx \left|\frac{1}{N}\sum_{n=1}^{N} \exp\left[j\frac{2\pi R}{\lambda}\Theta(\rho,\varphi) d_p \cos(\psi_n - \delta)\right]\right|, \tag{A3}
\end{align*}
where $\Delta \Phi_n(\rho,\varphi)= \Phi_n(\rho_2,\varphi_2)-\Phi_n(\rho_1,\varphi_1) $, $\Theta(\rho,\varphi) = \sqrt{\sin^2\rho\sin^2\varphi + \cos^2\rho\cos^2\varphi}$,  $d_p = \sqrt{(\rho_2-\rho_1)^2 + (\varphi_2-\varphi_1)^2}$ and $\delta$ is a phase offset dependent on the directions. 

For large $N$, approximate the sum by an integral 
\begin{align*}
    G_{\text{angle}} \approx \left|\frac{1}{2\beta}\int_{-\beta}^{\beta} 
e^{j kR\Theta d_p \cos(\psi-\delta)} \, d\psi\right|.\tag{A4}
\end{align*}
Utilize the Jacobi-Anger expansion \cite{jacobi}
\begin{align*}
    e^{jz\cos\theta} = \sum_{m=-\infty}^{\infty} j^m J_m(z) e^{jm\theta},\tag{A5}
\end{align*}
we obtain 
    \begin{align*}
        G_{\text{angle}} \approx & \left| J_0(\frac{2\pi R}{\lambda}\Theta d_p) + \frac{1}{\beta}\sum_{m=1}^{\infty} \frac{j^m}{m} \sin(m\beta) J_m(\frac{2\pi R}{\lambda}\Theta d_p) \right. \\
        & \left. \cdot (e^{-jm\delta}+(-1)^m e^{jm\delta}) \right|.\tag{A6}
    \end{align*}
For practical array sizes and moderate angular separations, higher-order Bessel terms decay rapidly. The dominant term is $m=0$, and a finite-aperture factor $  \frac{\sin\beta}{\beta}$ emerges. Thus, the beamforming gain in polar coordinates is approximately as
\begin{align*}
    G_{\text{angle}}(\mathbf{s}_1, \mathbf{s}_2) 
\approx \left| J_0\left( \frac{2\pi R}{\lambda} \cdot 
\Theta(\rho,\varphi) \cdot d_p \cdot \frac{\sin\beta}{\beta} \right) \right| \tag{A7}
\end{align*}
Finally, the proof is completes.

\section{}\label{app:proof_BG3}

Consider two beams at the same angular direction but different distances $r_1 \neq r_2$. According to \eqref{dis}, the phase difference at the $n$-th element can be expressed as
\begin{align*}
    \Delta\Phi_n = \frac{2\pi}{\lambda}(r_{n,2} - r_{n,1}) \approx \frac{\pi R^2}{\lambda}\left(\frac{1}{r_2} - \frac{1}{r_1}\right) F(\psi_n), \tag{B1}
\end{align*}
where $F(\psi_n) = 1 - 2\cos(\beta-\psi_n)\cos\beta + \cos^2\beta -  \sin\theta\sin\phi(\cos(\beta-\psi_n)-\cos\beta) + \cos\theta\sin(\beta-\psi_n)$.

Using trigonometric identities and converting to polar coordinates $(\rho,\varphi)$, after algebraic manipulation, we have
\begin{align*}
    F(\psi_n) = \xi(\rho,\varphi,\beta)\cos(2\psi_n - \alpha)+ C,\tag{B2}
\end{align*}
where $\xi(\rho,\varphi,\beta) = \sqrt{f_1^2(\beta,\rho,\varphi) + f_2^2(\beta,\rho,\varphi)}$ with $f_1(\rho,\varphi,\beta)  = \sin\rho\sin\varphi\sin\beta - \cos\rho\cos\beta$ and $f_2( \rho,\varphi, \beta)  = \cos\rho\cos\varphi\sin\beta$.

Then, the beamforming gain can be simplified to
\begin{align*}
    & G_r(\mathbf{s}_1,\mathbf{s}_2)= \left|\mathbf{b}^H(r_1,\rho,\varphi)\mathbf{b}(r_2,\rho,\varphi)\right| \\
    &\qquad= \left|\frac{1}{N}\sum_{n=1}^N e^{j\Delta\Phi_n}\right|\\
&\qquad \approx \left|\frac{1}{N}\sum_{n=1}^{N} \exp\left[j\frac{\pi R^2}{\lambda}\Delta\tau\,\xi(\rho,\varphi,\beta)\cos(2\psi_n - \alpha)\right]\right|,\tag{B3}
\end{align*}
where $\Delta\tau = \left|\frac{1}{r_2} - \frac{1}{r_1}\right|$. 
Similar to Appendix \ref{app:proof_BG1}, we obtain 
 \begin{align*}
     G_r(\mathbf{s}_1, \mathbf{s}_2) \approx \left|\frac{1}{N}\sum_{m=-\infty}^{\infty} j^m J_m(\zeta) e^{-jm\alpha} \sum_{n=1}^{N} e^{j2m\psi_n}\right|.\tag{B4}
 \end{align*}
where $\zeta = \frac{\pi R^2}{\lambda}\Delta\tau\,\xi(\rho,\varphi,\beta)$. 
\\
Since $ \frac{(n-1)L}{(N-1)R}$, we have 
\begin{equation}
\sum_{n=1}^{N} e^{j2m\psi_n} = 
\begin{cases}
N, & 2m = N \cdot t, \quad t \in \mathbb{Z} \\
0, & \text{otherwise}.\tag{B5}
\end{cases}
\end{equation}

For large $N$, using the asymptotic property $|J_{|m|}(\zeta)| \leq \left(\frac{\zeta e}{2|m|}\right)^{|m|}$, the dominant term is $m=0$. Thus 
\begin{align*}
   {G_r(\mathbf{s}_1,\mathbf{s}_2) \approx \left| J_0\left( \frac{\pi R^2 \xi(\rho,\varphi,\beta) |\Delta\tau|}{\lambda} \right) \right|}. \tag{B6}
\end{align*}
The approximation error is bounded by
\begin{align*}
    \Delta_G &\approx \left| j^N J_N(\zeta) \cdot 2\cos\left( \frac{\alpha N}{2} \right) \right| 
    \leq 2\left( \frac{\zeta e}{2N} \right)^N,\tag{B7}
\end{align*}
which decays monotonically with $N$, ensuring accurate approximation for large arrays. This completes the proof.

\section{}\label{app:proof_ERD}
According to \emph{Lemma 2}, the beamforming gain adapting the far field beam steering vector $\mathbf{a}(\rho,\varphi)$ at the location $(r,\rho ,\varphi)$ can be formulated as
\begin{align*}
    \mu(r,\rho,\varphi) &= \left|\mathbf{b}^H(r,\rho,\varphi)\mathbf{a}(\rho,\varphi)\right| = \left|\mathbf{b}^H(r,\rho,\varphi)\mathbf{b}(\infty,\rho,\varphi)\right| \\
    &\approx \left|J_0\left(\frac{\pi R^2 \xi( \rho, \varphi,\beta)}{\lambda r}\right)\right|. \tag{C1}
\end{align*}

According to the definition of the zero-order Bessel function inverse $J_0^{-1}(\cdot)$, the condition $1-\mu(r,\rho,\varphi) \ge \delta_{\text{gain}}$ is equivalent to $\frac{\pi R^2 \xi( \rho, \varphi,\beta)}{\lambda r} \ge |J_0^{-1}(1-\delta_{\text{gain}})|$. By rearranging this inequality, we can derive the lower bound of the ERD as 
\begin{align*}
    r_{\mathrm{ERD}}^L=\frac{\pi R^2 \xi(\rho, \varphi,\beta)}{\lambda |J_0^{-1}(1-\delta_{\text{gain}})|}, \tag{C2}
\end{align*}
where $\xi( \rho, \varphi,\beta)$ is the angular-dependent correction factor, $R$ denotes the aperture radius of the antenna array, and $\lambda$ is the wavelength of the transmitted signal. 
Moreover, it can be observed from the definition of $\mu(r,\rho,\varphi)$ that the beamforming gain is only dependent on the radial distance $r$ and is invariant with respect to the angular coordinates $(\rho,\varphi)$ under the far field approximation, thus the ERD can be simplified by omitting the angular variables $(\rho,\varphi)$. This completes the proof.   

\end{appendices}


\begin{thebibliography}{1}
\bibitem{6G}
C. You et al., ``Next generation advanced transceiver technologies for 6G,” \emph{IEEE J. Sel. Areas Commun.}, vol. 43, no. 3, pp. 582–627, Mar. 2025.
 
\bibitem{LYW}
Y. Liu et al., ``Near-Field Communications: A Comprehensive Survey," \emph{IEEE Commun. Surv. Tut.}, vol. 27, no. 3, pp. 1687-1728, Jun. 2025.

\bibitem{elaa1}
M. Fu et al., ``Extremely Large-Scale Movable Antenna-Enabled Multiuser Communications: Modeling and Optimization," \emph{IEEE Trans. Wirel. Commun.}, vol. 25, pp. 10290-10304, 2026,

\bibitem{elaa}
H. Lu et al., ``A Tutorial on Near-Field XL-MIMO Communications Toward 6G,” \emph{IEEE Commun. Surv. Tut.}, vol. 26, no. 4, pp. 2213-2257, Fourthquarter. 2024.

\bibitem{fresnel}
J. Sherman, ``Properties of focused apertures in the fresnel region,” \emph{IRE Trans. Antennas Propag.}, vol. 10, no. 4, pp. 399–408, Jul. 1962.

\bibitem{fnear}
K. T. Selvan et al., ``Fraunhofer and Fresnel distances: Unified derivation for aperture antennas,” \emph{IEEE Antennas Propag. Mag.}, vol. 59, no. 4, pp. 12–15, Aug. 2017.

\bibitem{region}
R. Li et al.,``Applicable regions of spherical and plane wave models for extremely large-scale array communications," \emph{China Commun.}, vol. 22, no. 5, pp. 128-151, May. 2025.

\bibitem{nfc}
Y. Liu et al., ``Near-Field Communications: A Tutorial Review," \emph{IEEE Open J. Commun. Soc.}, vol. 4, pp. 1999-2049, 2023.

\bibitem{nfc1}
Y. Guo et al., ``Modeling and Analysis of Spatial Correlation for Near-Field Communications," \emph{IEEE Trans. Commun.}, vol. 73, no. 12, pp. 13659-13676, Dec. 2025.

\bibitem{turntonfc}
Y. Le et al., ``Performance Analysis for Extremely Large-Scale MIMO Communication Systems," \emph{IEEE Commun. Lett.}, vol. 30, pp. 917-921, 2026.

\bibitem{turntonfc2}
W. Huang et al., ``Codebook design based on beam energy spread for extremely large-scale arrays,” \emph{IEEE Trans. Commun.}, early access, Aug. 11, 2025.

\bibitem{turntonfc3}
X. Zheng et al., ``Model-driven iterative super-resolution channel estimation for wideband near-field extremely large-scale MIMO systems,” \emph{IEEE Wireless Commun. Lett.}, vol. 14, no. 2, pp. 300–304, Feb. 2025.

\bibitem{beam}
H. Zhang et al., ``6G wireless communications: From far-field beam steering to near-field beam focusing,” \emph{IEEE Commun. Mag.}, vol. 61, no. 4, pp. 72–77, Apr. 2023.

\bibitem{cdbook}
Z. Yuan et al., ``Low-Complexity Codebook Design for Near-field Wideband Precoding and Its Experimental Validation," \emph{IEEE Trans. Antennas Propag.}, early access, 2026.

\bibitem{csi}
Z. Wang et al., ``Extremely large-scale MIMO: Fundamentals, challenges, solutions, and future directions,” \emph{IEEE Wireless Commun.}, vol. 31, no. 3, pp. 117–124, Jun. 2024.

\bibitem{focusing}
A. Kosasih et al., ``Finite beam depth analysis for large arrays,” \emph{IEEE Trans. Wireless Commun.}, vol. 23, no. 8, pp. 10015–10029, Aug. 2024.

\bibitem{focusing2}
Y. Guo et al., ``Wideband beamforming for near-field communications with circular arrays,” \emph{IEEE Trans. Wireless Commun.}, vol. 23, no. 12, pp. 19065–19082, Dec. 2024.

\bibitem{ula1}
X. Mestre et al., ``Beamfocusing Capabilities of a Uniform Linear Antenna Array in the Holographic Regime", \emph{IEEE Trans. Signal Proc.}, vol.73, pp.5015-5031, 2025.

\bibitem{ula}
M. Cui et al., ``Channel estimation for extremely large-scale MIMO: Far-field or near-field?” \emph{IEEE Trans. Commun.}, vol. 70, no. 4, pp. 2663-2677, Apr. 2022.

\bibitem{ula2}
X. Zhang et al., ``Near-Field Channel Estimation for Extremely Large-Scale Array Communications: A Model-Based Deep Learning Approach," \emph{IEEE Commun. Lett.}, vol. 27, no. 4, pp. 1155-1159, April 2023.

\bibitem{training}
X. Li et al., ``Codebook Design and Beam Training for Multi-User Modular XL-MIMO Communications: From Far-Field to Near-Field," \emph{IEEE Trans. Commun.}, vol. 73, no. 11, pp. 11855-11869, Nov. 2025.

\bibitem{enable}
Z. Wu et al., ``Enabling More Users to Benefit From Near-Field Communications: From Linear to Circular Array," \emph {IEEE Trans. Wireless Commun.}, vol. 23, no. 4, pp. 3735-3748, Apr. 2024.

\bibitem{uca}
Y. Xie et al., ``Near-field beam training in THz communications: The merits of uniform circular array,” \emph{IEEE Wireless Commun. Lett.}, vol. 12, no. 4, pp. 575-579, Apr. 2023.

\bibitem{uca2}
M. Parvini et al., ``Interference Characterization and Mitigation in Near-Field XL-MIMO: A Case for Linear and Circular Array Geometries," \emph{IEEE Wireless Commun. Lett.}, vol. 15, pp. 111-115, 2026.

\bibitem{training2}
B. Qin et al., ``Fast Distance Sampling for Uniform Circular Array in Near-Field 3D Beam-Focusing," \emph{IEEE Signal Process. Lett.}, vol. 31, pp. 2210-2214, 2024.

\bibitem{cylin}
D. G. Riviello et al., ``Multi-Layer Multi-User MIMO With Cylindrical Arrays Under 3GPP 3D Channel Model for B5G/6G Networks," \emph{IEEE Access}, vol. 12, pp. 145753-145767, 2024.

\bibitem{cla}
X. Wang et al., ``Near-field codebook design for extremely large cylindrical antenna array systems,” \emph{IEEE Trans. Commun.}, vol. 72, no. 9, pp. 5380-5395, Sep. 2024.

\bibitem{UAV}
J. Yang et al., ``A Cylindrical Phased Array Radar System for UAV Detection," in \emph{Proc. 6th Int. Conf. Intell. Comput. Signal Process. (ICSP)}, Apr. 2021, pp. 894–898.

\bibitem{Maritime}
S. Sun et al., "Modeling and Analysis of Land-to-Ship Maritime Wireless Channels at 5.8 GHz," in \emph{IEEE Trans. Wireless Commun.}, vol. 25, pp. 10051-10065, 2026.

\bibitem{de}
S. Droulias et al., ``Orthogonal Beams and Codebook Design for Near-Field Space-Division Multiple Access in 6G Systems," \emph{IEEE Trans. Antennas Propag.}, vol. 73, no. 9, pp. 6899-6913, Sept. 2025.

\bibitem{caa}
E. Crespo-Bardera et al.,  ``Design and analysis of conformal antenna for future public safety communications: Enabling future public safety communication infrastructure,” \emph {IEEE Antennas Propag. Mag.}, vol. 62, no. 4, pp. 94-102, Aug. 2020.

\bibitem{aulaupa}
L. Xue et al.,  ``Spatial Correlation and Degrees of Freedom in Arched HMIMO Arrays: A Closed-Form Analysis," \emph{2025 IEEE 101st Vehicular Technology Conference (VTC2025-Spring)}, Oslo, Norway, 2025, pp. 1-6.

\bibitem{polarcd}
A. Abdallah et al.,  ``Exploring frontiers of polar-domain codebooks for near-field channel estimation and beam training: A comprehensive analysis, case studies, and implications for 6G,” \emph{IEEE Signal Process. Mag.}, vol. 42, no. 1, pp. 45–59, Jan. 2025.


\bibitem{3Ds}
C. Yang et al.,  ``Near-Field Codebook-Based 3-D Spherical Channel Estimation for UCA XL-MIMO Systems," \emph{IEEE Wireless Commun. Lett.}, vol. 14, no. 10, pp. 3154-3158, Oct. 2025.
 

\bibitem{rrd}
S. Sun et al.,  ``How to differentiate between near field and far field: Revisiting the Rayleigh distance." \emph{IEEE Commun. Mag.}, vol. 63, no. 1, pp. 22–28, Jan. 2025.

\bibitem{pi/8}
C. A. Balanis, \emph{Antenna Theory: Analysis and Design}. Hoboken, NJ, USA: Wiley, 2016.

\bibitem{DFT}
X. Wu et al.,  ``Near-field beam training with DFT codebook,” \emph{ IEEE Wireless Commun. Netw. Conf. (WCNC)}, Apr. 2024, pp. 1–6.

\bibitem{jacobi}
F. Bowman, \emph{Introduction to Bessel Functions}. North Chelmsford, U.K.: Courier Corporation, 2012.







\end{thebibliography}
\end{document}